\documentclass[sigconf,nonacm]{acmart}

\usepackage{enumitem}
\usepackage{mathtools}
\usepackage[ruled, lined, linesnumbered]{algorithm2e}
\SetKwInput{KwInput}{Input}
\SetKwInput{KwReturn}{Return}
\allowdisplaybreaks

\usepackage{tikz}
\newcommand*\circled[1]{\tikz[baseline=(char.base)]{\node[shape=circle,draw,inner sep=0.05pt] (char) {#1};}}

\begin{document}

\title[Analytical Power-Aware Provisioning for Prefill--Decode Disaggregated AI Inference]
{Analytical Power-Aware Provisioning\texorpdfstring{\\}{ }for Prefill--Decode Disaggregated AI Inference}

\author{Mingyuan Yan}
\email{my3188@nyu.edu}
\affiliation{%
  \department{Department of Electrical and Computer Engineering}
  \institution{New York University}
  \city{Brooklyn}
  \state{New York}
  \country{USA}
}

\author{Haiyu Wang}
\email{hw3689@nyu.edu}
\affiliation{%
  \department{Department of Electrical and Computer Engineering}
  \institution{New York University}
  \city{Brooklyn}
  \state{New York}
  \country{USA}
}

\author{Linxuan Biao}
\email{lb4770@nyu.edu}
\affiliation{%
  \department{Department of Electrical and Computer Engineering}
  \institution{New York University}
  \city{Brooklyn}
  \state{New York}
  \country{USA}
}

\author{H. Jonathan Chao}
\email{chao@nyu.edu}
\affiliation{%
  \department{Department of Electrical and Computer Engineering}
  \institution{New York University}
  \city{Brooklyn}
  \state{New York}
  \country{USA}
}

\author{Sai Qian Zhang}
\email{sai.zhang@nyu.edu}
\affiliation{%
  \department{Department of Electrical and Computer Engineering}
  \institution{New York University}
  \city{Brooklyn}
  \state{New York}
  \country{USA}
}

\author{Wenqi Cui}
\authornote{Corresponding author.}
\email{wenqicui@nyu.edu}
\affiliation{%
  \department{Department of Electrical and Computer Engineering}
  \institution{New York University}
  \city{Brooklyn}
  \state{New York}
  \country{USA}
}

\begin{abstract}

Power availability increasingly constrains the operation of AI inference fleets, creating a need for provisioning methods that jointly consider serving capacity and power consumption.
Prefill--decode (PD) disaggregation has emerged as a prevalent architecture for large-scale inference serving.
However, determining the appropriate numbers of prefill and decode instances is challenging because serving capacity depends jointly on workload characteristics, hardware constraints, queueing, and KV-cache reservations.
Existing approaches largely rely on profiling and simulation, providing limited analytical insight into how provisioning decisions shape the tradeoff between serving capacity and power consumption.

This paper develops an analytical framework for power-aware provisioning of PD-disaggregated AI inference.
Given an inference workload and hardware,
the framework models the serving capacity and average power consumption of a provisioned deployment. The serving-capacity model is derived from the joint distribution of input--output lengths and hardware compute and memory limits. In particular, it explicitly captures the coupling between prefill and decode induced by KV-cache reservations, as well as the impact of request queueing.
On this basis, the power model determines per-instance power consumption as a function of normalized serving throughput.
Together, the models determine the serving capacity--power Pareto front among candidate provisioned deployments,
enabling the service provider to choose a provisioned deployment as the workload or available power changes.

\end{abstract}

\maketitle

\section{Introduction}
\label{sec:intro}
Electricity demand from datacenters is growing rapidly
and is expected to add hundreds of terawatt-hours to annual electricity demand,
with AI accounting for a substantial share of this growth~\cite{iea2026keyquestions}.
Expanding electricity generation and grid capacity, however,
typically takes much longer than constructing a datacenter~\cite{iea2025europedc,chen2025griddemand}.
As a result, datacenter projects increasingly face delayed
or denied grid connections~\cite{nyeo62,txdirective2026},
as well as conditional interconnection agreements that limit
when and how much power a datacenter may consume from the grid~\cite{cru2025leu}.
These constraints make two properties of computing loads increasingly valuable:
efficiency in delivering a required service while consuming as little power as possible,
and flexibility in reducing power demand while maintaining as much serving capacity as possible.

Within AI workloads,
inference is becoming a major and persistent source of electricity demand~\cite{patterson2022carbon}.
Both request volume and request size are increasing:
large language model (LLM) based services are becoming routine~\cite{chatterji2025chatgpt,deepseekinfra2025},
agentic workflows turn one user task into multiple model calls~\cite{yao2023react,copilottraces2026},
and reasoning models use more tokens before producing an answer~\cite{deepseekr1}.
These trends make power-aware inference provisioning increasingly important,
both to reduce the power required for a given serving capacity
and to understand how much serving capacity can be maintained as power is reduced.

An AI inference request is typically processed in two phases: \emph{prefill} and \emph{decode}.
During \emph{prefill}, the model processes the entire input sequence,
constructs the key-value (KV) cache reused during decode, and generates the first output token.
During \emph{decode},
the model reuses that cache and generates the remaining output one token at a time~\cite{orca2022}.
Prefill is highly parallel and generally compute-bound,
whereas decode repeatedly reads model weights and KV-cache data and is generally limited by memory bandwidth.
This difference motivates prefill--decode (PD) disaggregation~\cite{splitwise2024,distserve2024}, which places the two phases in separate pools of model instances that can be provisioned independently.
This architecture has consequently been widely used in a range of large-scale inference systems~\cite{mooncake2025,deepseekinfra2025,rtpllm2026,cloudmatrix2025,heteroscale2025,metavllm2025}.

Choosing how many prefill instances
and how many decode instances to run is a central provisioning decision for PD-disaggregated AI inference.
For a given workload and hardware platform,
this choice determines the system's power consumption and its serving capacity (i.e., the maximum request rate it can sustain).
Overprovisioning provides serving capacity beyond the required level but incurs additional resource and power costs,
whereas underprovisioning fails to meet the required serving capacity.
The two pools must also be balanced:
insufficient serving capacity in either phase creates a bottleneck
and can leave resources in the other phase underutilized.
These tradeoffs motivate a joint characterization of serving capacity
and power across provisioned deployments.

Existing PD provisioning methods commonly evaluate candidate deployments using simulation~\cite{splitwise2024,distserve2024,bestserve2025,frontier2026} or empirical performance measurements~\cite{pdallocation2026}.
However, hardware profiling is time-consuming, while simulation-based comparisons must be repeated when the request-length distribution changes.
Analytical models reduce this dependence on empirical evaluation~\cite{distserve2024,bestserve2025,niequeueing2026,fluidwait2025,afdprovisioning2026}, but existing approaches
do not capture the coupling between prefill and decode in a PD-disaggregated deployment.
Consequently, they cannot characterize deployment-level serving capacity, and load-dependent power consumption is generally not modeled.
This motivates the central question of this paper: \emph{Can we develop analytical models for PD-disaggregated AI inference that characterize how provisioning decisions shape the tradeoff between serving capacity and power?}

We therefore develop analytical models of serving capacity and average deployment power that apply across different workload distributions and hardware platforms.
The key challenge in characterizing serving capacity is that prefill activity changes the cache memory available for decode.
Each request reserves KV-cache space on its assigned decode instance before prefill begins.
As a result, the decode-side KV-cache pool is shared by active decode requests
and requests waiting for or undergoing prefill.
This coupling limits the number of active decode requests,
thereby affecting decode capacity and the characterization of serving capacity.

To this end, we develop a serving-capacity model from the joint distribution of input
and output lengths and the hardware compute and memory limits.
The model accounts for the KV-cache space reserved while requests wait for or undergo prefill,
which reduces the memory available to active decode requests.
We model prefill and decode capacities separately and combine their aggregate capacities
to obtain the end-to-end serving capacity.
We then develop a power model
that relates each instance's power consumption to its normalized serving throughput.
Together, the models determine the serving capacity--power Pareto front among candidate provisioned deployments.
The front shows the minimum-power provisioned deployment for a required serving capacity
and how much serving capacity can be maintained under a reduced power budget.
We evaluate the models across different workloads, including workloads derived from the Mooncake~\cite{mooncake2025} and Azure~\cite{splitwise2024} production traces.
The main contributions of the paper are summarized below:
\begin{itemize}[leftmargin=10pt]
    \item \textbf{Analytical serving-capacity model.}
    We develop an analytical model of the serving capacity of PD-disaggregated inference
    from the joint distribution of request input and output lengths, as well as hardware compute and memory limits.
    The model explicitly captures the coupling between prefill and decode induced by KV-cache reservations, as well as the impact of request queueing.
    \item \textbf{Load-dependent power model.}
    We develop separate power models for prefill and decode instances
    as functions of normalized serving throughput.
    The model captures how the provisioned numbers of instances and the serving load jointly determine the average power consumption of the deployment.
    \item \textbf{Serving capacity--power Pareto front.}
    Combining the serving-capacity and power models yields a discrete Pareto front over candidate provisioned deployments. This characterization quantifies the tradeoff between power consumption and serving capacity, providing a principled basis for grid-friendly inference control and provisioning decisions.
    \item \textbf{Provisioning for efficiency and flexibility.}
    We use the Pareto front to identify minimum-power provisioned deployments
    for a required serving capacity
    and to quantify how much serving capacity can be maintained as available power is reduced.
    We evaluate these provisioning decisions across different workloads
    and candidate deployments against measurements.
\end{itemize}

\subsection{Related Work}
\label{sec:related}

This work is closely related to topics on provisioning for
PD-dis\-ag\-gre\-gated inference,
analytical models, and power-aware control.

\noindent\paragraph{PD Disaggregation and Serving-Capacity Models.}
Under PD disaggregation~\cite{splitwise2024,distserve2024}, the serving capacity of AI inference depends on the numbers of prefill and decode instances and the request workload.
Existing provisioning methods commonly evaluate candidate deployments through simulation~\cite{splitwise2024,distserve2024,bestserve2025,frontier2026} or throughput measurements~\cite{pdallocation2026}.
However, hardware profiling~\cite{distserve2024,bestserve2025} can be time-consuming, while simulation-based comparisons must be repeated when the workload distribution changes.

To reduce reliance on purely empirical evaluation, several works introduce analytical models for inference performance. DistServe~\cite{distserve2024} and BestServe~\cite{bestserve2025} develop analytical models for
batch execution times, but still rely on simulation to determine the request rates that meet latency requirements.
A queueing-based method~\cite{pdallocation2026} derives the supported prefill request rate by applying an analytical queueing model to measured prefill performance, whereas it obtains decode throughput directly from measurements.
There also exist other analytical models~\cite{niequeueing2026,fluidwait2025,afdprovisioning2026} for architectures that are not PD-disaggregated, but they cannot be directly applied to capture the coupling between the prefill and decode pools. In addition, these models alone are not sufficient to characterize the tradeoff between serving capacity and power. To address these gaps, this work develops analytical models for both PD serving capacity and power, while accounting for prefill--decode coupling induced by KV-cache reservations and the impact of request queueing.

\paragraph{Power Modeling.}
Analytical power models for PD-disaggregated AI inference remain limited. Splitwise~\cite{splitwise2024} determines the numbers of prefill and decode instances under a power budget, but models deployment power as proportional to the number of provisioned instances, without accounting for the effect of workload on power consumption. Several studies model the relationship between power consumption and batch size~\cite{g2g2026,onewlaw2026,inferencefleetsim2026}.
Other approaches estimate average GPU power using kernel-level information~\cite{energaizer2026} or publicly available model and GPU specifications~\cite{wattgpu2026}.
However, none of these models captures how the numbers of prefill and decode instances, together with the workload served by each pool, jointly determine the average power consumption of a PD-disaggregated deployment. A separate line of work models energy per request or token~\cite{wilkins2024offline,sweetspot2026,tokenswatthours2026,tokenpowerbench2026,wattcounts2026}; however, such energy models do not directly characterize the aggregate power demand of an inference deployment. In contrast, this paper models prefill and decode power separately as functions of normalized serving throughput, thereby capturing how workload-dependent utilization of the two pools translates into deployment-level power consumption.

\paragraph{Power-Aware Control of Inference Workloads.}
Runtime power control typically starts from an already provisioned inference deployment and adjusts its operating configuration in response to workload or power conditions.
Some methods keep the number of instances fixed and control GPU frequency~\cite{voltanallm2025,greenllm2025},
batch size~\cite{pals2026}, or quantization~\cite{quantdemandresponse2026}.
For PD disaggregation, runtime systems adjust the numbers of prefill and decode instances separately,
to reduce energy consumption~\cite{dualscale2026}, meet power budgets~\cite{rapid2026,powerslider2026},
or respond to feedback from the running fleet~\cite{arrow2025,metavllm2025}. Other approaches regulate the aggregate power demand through coordination with external energy resources, such as battery energy storage systems, uninterruptible power supplies (UPSs), and supercapacitors~\cite{ko2025mitigation,xie2026data,koksal2026theory}. In contrast, this paper focuses on the provisioning configuration itself, using the numbers of prefill and decode instances as the control variables. This choice requires an analytical characterization of how each candidate provisioning configuration determines both serving capacity and power consumption, which enables systematic comparison among candidate deployments.

\section{Problem Formulation}
\label{sec:optdep}
\subsection{Provisioning and Workload}
\label{sub:deploy}

An inference request contains an input sequence of $\ell_\mathrm{in}$ tokens
and produces an output sequence of $\ell_\mathrm{out}$ tokens.
AI inference proceeds in two phases.
The prefill phase processes the complete input, constructs the KV cache, and produces the first output token.
The decode phase then produces each remaining token in a separate forward pass.
Under PD disaggregation, the two phases run on separate pools of model instances.
Here, an \emph{instance} is one running replica of the model on its assigned hardware.
The service provider provisions $n_P$ prefill instances and $n_D$ decode instances.
Together, these two numbers define a \emph{provisioned deployment} $(n_P,n_D)$.
We label deployments as $n_P\mathrm{p}n_D\mathrm{d}$; for example, 3p2d denotes a deployment with three prefill instances and two decode instances.

The numbers of prefill and decode instances required depend on the inference workload.
We denote the workload by $\mathcal{W}:=(\mathcal{L},\lambda)$.
It consists of the joint distribution $\mathcal{L}$ of the input
and output lengths $(\ell_\mathrm{in},\ell_\mathrm{out})$,
and the mean request arrival rate $\lambda$.
The input length determines the computational work in prefill,
while both the input and output lengths determine the work in decode.
We assume that request arrivals are Poisson and that the workload remains stationary within each provisioning period.
When the arrival rate or the length distribution changes,
the provisioned deployment is chosen again for the new workload.

The provisioned deployment needs to satisfy the service-quality requirement,
which is defined by service level objectives (SLOs)
and is commonly a bound on the latency its users experience.
The SLOs mainly bound time to first token (TTFT),
the time from the arrival of a request to its first output token, and time per output token (TPOT),
the time between consecutive output tokens.
A provisioned deployment meets its SLOs when both TTFT and TPOT stay below their bounds.

Together, the workload $\mathcal{W}$
and the provisioned deployment $(n_P,n_D)$ determine the inference system's serving capacity,
its power consumption, and the latency its users experience.
Provisioning chooses $(n_P,n_D)$ for a given workload so that the fleet meets its SLOs
while consuming as little power as possible.

\subsection{Power-Aware Provisioning}
\label{sub:problem}
The service provider seeks the provisioned deployment that provides the required serving capacity
while consuming as little power as possible.
For a provisioned deployment $(n_P,n_D)$ under a workload $\mathcal{W}$,
let $\mathcal{P}(n_P,n_D;\mathcal{W})$ denote the average power of that provisioned deployment over a long serving window.
Its \emph{serving capacity} $\mu(n_P,n_D;\mathcal{W})$ is the maximum request rate it can sustain~\cite{vidur2024}.
To provide the required headroom, the provisioned deployment must have a serving capacity of at least $\lambda_{\min}$.
The resulting provisioning problem is
\begin{equation}
\begin{aligned}
\min_{n_P,n_D\in\mathbb{Z}_{+}}\quad & \mathcal{P}(n_P,n_D;\mathcal{W}) \\
\text{s.t.}\quad & \mu(n_P,n_D;\mathcal{W}) \;\ge\; \lambda_{\min} .
\end{aligned}
\label{eq:slo}
\end{equation}
We model average power rather than instantaneous power,
since short-term power fluctuations of individual instances tend to average out across a large provisioned deployment.
Peak-power behavior is outside the scope of this paper.
A provisioned deployment that satisfies this constraint is \emph{feasible}.
The constraint compares the serving capacity of a provisioned deployment with the required serving capacity of the workload and its SLOs.
A change in the workload $\mathcal{W}$ requires solving~\eqref{eq:slo} again.

The required serving capacity $\lambda_{\min}$ is chosen based on the current request arrival rate
and the headroom it requires to handle fluctuations in arrivals and satisfy the SLOs.
We therefore define $ \lambda_{\min}=\frac{\lambda}{\bar\rho}, $ where $0<\bar\rho<1$ is the maximum serving utilization chosen by the service provider.
A lower $\bar\rho$ reserves more headroom for tighter SLOs and variations in arrivals.
Prior work uses values of $0.85$ and $0.90$~\cite{fleetopt2026,inferencefleetsim2026,niequeueing2026}.
The parameter $\bar\rho$ can also serve as a control parameter for power-responsive operation.
When grid power is constrained,
the service provider can temporarily reduce the reserved headroom by increasing $\bar\rho$,
thereby lowering the required serving capacity $\lambda_{\min}$ and enabling operation with fewer instances
and lower power consumption,
provided that the resulting service degradation remains within short-term tolerance.

\subsection{Serving Capacity--Power Pareto Front}
\label{sub:pareto}
\begin{figure}[tb]
  \centering
  \includegraphics[width=\columnwidth]{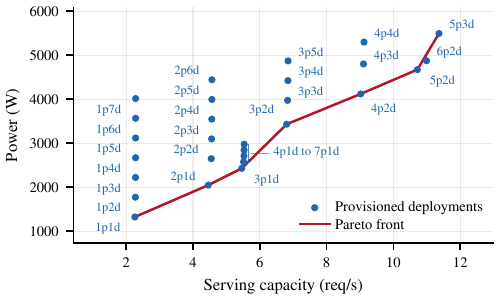}
  \vspace{-0.3cm}
  \caption{Modeled serving capacity and power of candidate deployments under a fixed-length workload, with the corresponding Pareto front.
  Points are labeled by deployment.}
  \label{fig:pareto}
  \vspace{-0.3cm}
\end{figure}

Although the provisioning problem~\eqref{eq:slo} is an integer optimization problem, its decision variables consist only of the numbers of prefill and decode instances, $(n_P,n_D)$. This low-dimensional structure makes it unnecessary to develop a dedicated integer optimization algorithm. Instead, the analytical models developed in Sections~\ref{sec:model} and~\ref{sec:derivation} can be used to obtain a discrete Pareto front between power and serving capacity, where the provisioning problem can then be reduced to a search over the resulting discrete Pareto front.

Specifically, each candidate $(n_P,n_D)$ corresponds to a point in the serving capacity--power plane, as illustrated in Figure~\ref{fig:pareto}. A deployment is Pareto-optimal if no other deployment provides at least the same serving capacity with lower power consumption. The collection of such deployments forms the discrete Pareto front, which characterizes the minimum power required to achieve different levels of serving capacity.
Among the Pareto-optimal deployments satisfying the serving-capacity requirement $\lambda_{\min}$, the one with the lowest power consumption is selected as the optimal provisioned deployment for solving~\eqref{eq:slo}. Thus, despite its integer formulation, the two-dimensional provisioning problem admits a simple geometric solution.

However, the Pareto front may change with the operating conditions, particularly with the request workload distribution. For a fixed workload distribution, an empirical Pareto front can be constructed from measurements. Such empirical profiling, however, is time-consuming, and the number of candidate deployments increases rapidly with the number of provisioned instances. Repeating this process for every possible workload distribution is therefore impractical. To address this challenge, we develop analytical models for the power model $\mathcal{P}(n_P,n_D;\mathcal{W})$ and the serving-capacity model $\mu(n_P,n_D;\mathcal{W})$. Given an online workload distribution $\mathcal{W}$, these models directly evaluate the power and serving capacity of each candidate deployment, allowing the corresponding Pareto front to be constructed without repeated empirical profiling. Figure~\ref{fig:pareto} shows the serving capacity--power plane the models determine for one workload. In the next sections, we present these two models and derive their analytical forms.

\section{Analytical Models}
\label{sec:model}
This section develops analytical models that relate a workload and a provisioned deployment to its serving capacity and average power.
We first model the serving capacities of the prefill and decode pools, then characterize the KV-cache constraint that determines the operating batch, and finally derive the deployment-level serving capacity and power.
Section~\ref{sec:derivation} provides the detailed derivations and calibration procedures.

\subsection{Serving Capacity Model}
\label{sub:capacity}
Each request runs first on a prefill instance and then on a decode instance,
with its KV-cache transferred between them.
Modern KV-transfer backends~\cite{mooncake2025,splitwise2024,tokenscale2025} perform this transfer asynchronously and overlap much of it with prefill computation.
We therefore assume that KV-cache transfer does not limit serving capacity and focus on the prefill and decode pools.
Let $\mu_P$ and $\mu_D$ denote the serving capacities of each prefill and decode instance, respectively. A deployment with $n_P$ prefill instances and $n_D$ decode instances therefore provides aggregate capacities of $n_P\mu_P$ and $n_D\mu_D$. Since every request passes through both pools, the end-to-end serving capacity is limited by the bottleneck pool:
\begin{equation}
\mu = \min\!\big(n_P\,\mu_P,\; n_D\,\mu_D\big).
\label{eq:twobottleneck}
\end{equation}

Prefill and decode are governed by different hardware bottlenecks.
Prefill processes all input tokens in a single forward pass and is typically \emph{compute-bound}, whereas decode generates one token per active request per iteration and repeatedly accesses model weights and KV-cache data, making it typically \emph{memory-bandwidth-bound}.
We therefore model prefill capacity from compute throughput and decode capacity from memory bandwidth.

\begin{table}[tb]
\centering
\caption{Serving-capacity model characterization.}
\vspace{-0.2cm}
\label{tab:notation}
\footnotesize\setlength{\tabcolsep}{10pt}
\begin{tabular}{lll}
\toprule
 & Prefill & Decode \\
\midrule
resource demand  & $F$ (ops)       & $Q$ (bytes) \\
peak performance & $\pi$ (ops/s)  & $\beta$ (bytes/s) \\
utilization      & MFU                    & MBU \\
processing time  & $t_P$                  & $t_D$ \\
\midrule
modeling unit       & request              & batch iteration \\
\# requests / unit  & $1$                  & $B$ \\
\# tokens / request & $\ell_\mathrm{in}$ & $1$ \\
\bottomrule
\vspace{-0.5cm}
\end{tabular}
\end{table}

We use \emph{processing time} as the intermediate variable for deriving serving capacity.
For prefill, the processing time is the service time $t_P$ of a single request, whereas for decode, it is the duration $t_D$ of one batch iteration, during which each active request generates one output token.
Despite this difference in interpretation, both phases follow the same roofline principle:
\[
\text{processing time}
=
\frac{\text{resource demand}}
{\text{peak performance}\times\text{utilization}}.
\]
Serving capacity then follows by converting this processing time from the phase-specific modeling unit in Table~\ref{tab:notation} to requests per second.

For compute-bound prefill, the resource demand is measured by the number of floating-point operations, denoted by $F$ (ops). The peak compute throughput, denoted by $\pi$ (ops/s), is a hardware parameter determined by the accelerator architecture and numerical precision. We represent the achieved fraction of peak compute throughput by an effective model FLOPs utilization (MFU), which aggregates the execution efficiency of linear-layer and attention computation. The resulting prefill service time is $t_P = F/(\pi\,\mathrm{MFU})$.

For bandwidth-bound decode, the resource demand is measured by the amount of
data transferred from memory, denoted by $Q$ (bytes).
The peak memory bandwidth, denoted by $\beta$ (bytes/s), is likewise a hardware
parameter.
We represent the achieved memory bandwidth by an effective memory bandwidth utilization (MBU), which aggregates the effects of model-weight and KV-cache accesses. The resulting decode-iteration time is $t_D = Q/(\beta\,\mathrm{MBU})$.

Table~\ref{tab:notation} summarizes characterizations for the two models.
Although prefill and decode share the same resource-based timing principle, they require different conversions from processing time to serving capacity, as derived next.

\subsubsection{Serving Capacity of a Prefill Instance}
\label{sub:mup}

Prefill is compute-bound~\cite{pope2023scaling,sarathiserve2024}, so we quantify its processing demand in FLOPs.
Because matrix multiplications account for most of prefill computation,
the demand $F_r$ for request $r$ with $\ell_{\mathrm{in},r}$ input tokens is
\begin{equation}
F_r = \underbrace{2N\,\ell_{\mathrm{in},r}}_{\text{linear layers}}
  + \underbrace{c_a Ld\,\ell_{\mathrm{in},r}^2}_{\text{attention}} ,
\label{eq:flop}
\end{equation}
where $N$ is the number of model parameters, $L$ is the number of transformer layers, and $d$ is the attention width.
The linear term approximates the dense linear-layer cost as one multiplication and one addition per parameter per input token.
The quadratic term comes from causal attention, and $c_a$ depends on how the attention kernel applies the causal mask.

Under the roofline model~\cite{williams2009roofline},
the prefill instance runs at an achieved compute rate of $\pi\,\mathrm{MFU}$.
Dividing this FLOP demand by the achieved compute rate $\pi\,\mathrm{MFU}$ gives the prefill service time $t_{P,r}$:
\begin{equation}
\begin{aligned}
t_{P,r} &= \frac{F_r}{\pi\,\mathrm{MFU}} \\[2pt]
    &= \underbrace{\frac{2N}{\pi\,\mathrm{MFU}}}_{=: a_P}\,\ell_{\mathrm{in},r}
     + \underbrace{\frac{c_a Ld}{\pi\,\mathrm{MFU}}}_{=:  b_P}\,\ell_{\mathrm{in},r}^2 .
\end{aligned}
\label{eq:tp}
\end{equation}

A prefill instance's serving capacity is the reciprocal of its mean service time over requests:
\begin{equation}
\begin{aligned}
\frac{1}{\mu_P}
&= \mathbb{E}_r[t_P] \\
&= a_P\,\mathbb{E}_r[\ell_\mathrm{in}] + b_P\,\mathbb{E}_r[\ell_\mathrm{in}^2] .
\end{aligned}
\label{eq:mup}
\end{equation}
where $\mathbb{E}_r[\cdot]$ denotes the expectation over requests drawn from the length distribution $\mathcal{L}$.
Prefill capacity therefore depends on the first two moments of the input-length distribution.

\subsubsection{Serving Capacity of a Decode Instance}
\label{sub:mud}
At typical serving batch sizes, decode is memory-bound~\cite{pope2023scaling,yuan2024roofline}. Modern LLM serving systems commonly use continuous batching, in which completed requests leave the decode batch and requests that complete prefill are admitted to it between decode iterations. Consequently, the number of active requests varies over time, so decode capacity depends jointly on batching dynamics, request queueing, and per-request processing.
We show the resulting decode-capacity model below and elaborate its derivation in Section~\ref{sub:der_mud}.

Because decode is memory-bandwidth-bound, we model its processing demand in terms of memory traffic.
We define the operating batch $B$ as the mean number of active requests across decode iterations.
In each iteration, the decode instance reads the model weights once, requiring $2N$ bytes in bf16,
and accesses the KV cache of every active request, with $\kappa$ bytes transferred per context token.
The mean memory traffic $\bar{Q}$ per decode iteration is therefore
\begin{equation}
\bar{Q} = \underbrace{2N}_{\text{model weights}}
  + \underbrace{\kappa B\bar{\ell}_{\mathrm{ctx}}}_{\text{KV cache}} ,
\label{eq:qmean}
\end{equation}
where $\bar{\ell}_{\mathrm{ctx}}$ denotes the mean context length of an active request.

The mean $\bar{\ell}_{\mathrm{ctx}}$ differs from the mean context length of arriving requests: requests with longer outputs remain active for more decode iterations and are therefore sampled more often.
Averaging over active request--iteration pairs gives

\begin{equation}
\bar{\ell}_{\mathrm{ctx}}
= \frac{\mathbb{E}_r\!\left[
    (\ell_\mathrm{out}-1)\left(\ell_\mathrm{in}+\frac{\ell_\mathrm{out}}{2}\right)
  \right]}
  {\mathbb{E}_r\big[\ell_\mathrm{out}-1\big]} .
\label{eq:theta}
\end{equation}
The mean decode-iteration time $\bar{t}_D$ is the mean memory traffic $\bar{Q}$ divided by the achieved memory bandwidth $\beta\,\mathrm{MBU}$:
\begin{equation}
\begin{aligned}
\bar{t}_D &= \frac{\bar{Q}}{\beta\,\mathrm{MBU}} \\[2pt]
          &= \underbrace{\frac{2N}{\beta\,\mathrm{MBU}}}_{=: a_D}
           + \underbrace{\frac{\kappa\bar{\ell}_{\mathrm{ctx}}}{\beta\,\mathrm{MBU}}}_{=: b_D}\,B .
\end{aligned}
\label{eq:td_mean}
\end{equation}

We next convert iteration time into request-level serving capacity.
Each decode iteration generates one token for each active request and therefore produces $B$ output tokens on average.
Since prefill generates the first output token, each request requires $\mathbb{E}_r[\ell_\mathrm{out}-1]$ decode tokens on average.
The mean decode time per request is therefore
\begin{equation}
\frac{1}{\mu_D} = \frac{\mathbb{E}_r[\ell_\mathrm{out}-1]\,\bar{t}_D}{B} .
\label{eq:mud}
\end{equation}

The remaining unknown is the operating batch $B$, which is constrained by the available KV-cache memory. The next subsection outlines the model used to characterize $B$.

\subsubsection{Operating Batch of a Decode Instance}
\label{sub:operating_batch}
One key challenge is that the operating batch $B$ is not a predetermined parameter.
Modern serving systems commonly employ continuous batching, where the instantaneous batch size is determined dynamically to utilize the KV-cache capacity of a decode instance while reserving sufficient space for requests undergoing prefill to be admitted to decode upon completion~\cite{mooncake2025}. For a request with input length $\ell_\mathrm{in}$, we denote this reservation by
$\ell_\mathrm{in}+R$ token slots, where $R$ is the additional space reserved for subsequently generated output tokens.

\begin{figure*}[tb]
  \centering
  \includegraphics[width=\textwidth]{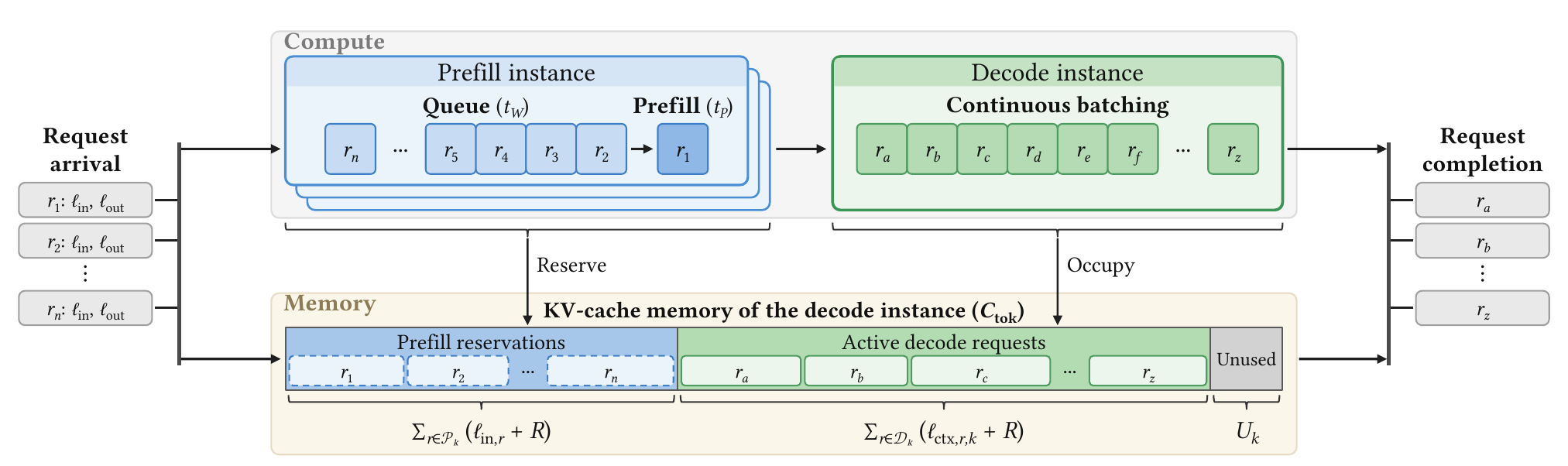}
  \caption{Request processing and decode-side KV-cache allocation under PD disaggregation.
  Requests waiting for or undergoing prefill reserve space in the same KV-cache pool used by active decode requests.}
  \label{fig:memory_balance}
\end{figure*}

Prefill reservations reduce the KV-cache space available to active decode requests (Figure~\ref{fig:memory_balance}).
The operating batch is therefore determined by both the KV-cache capacity and the reservations held during prefill.
Let $C_\mathrm{tok}$ denote the total KV-cache capacity of a decode instance.
Averaged over iterations, this capacity is divided among three components:
(i) reservations held by requests waiting for or undergoing prefill,
(ii) KV-cache space occupied by requests actively decoding, and
(iii) a small amount of unused capacity when the next reservation cannot fit.
Their mean occupancies satisfy

\begin{equation}
\begin{aligned}
C_\mathrm{tok} &\approx
\underbrace{\mu_D\left(
\bar t_W\,\mathbb{E}_r[\ell_\mathrm{in}+R]
+ \mathbb{E}_r\!\left[t_P(\ell_\mathrm{in}+R)\right]
\right)}_{\text{prefill occupancy}} \\[4pt]
&\quad + \underbrace{B\big(\bar{\ell}_{\mathrm{ctx}}+R\big)}_{\text{decode occupancy}}
\;+\;\underbrace{\frac{\mathbb{E}_r[(\ell_\mathrm{in}+R)^2]}{2\,\mathbb{E}_r[\ell_\mathrm{in}+R]}}_{\text{unused capacity}} .
\end{aligned}
\label{eq:operating_batch}
\end{equation}
where $\bar t_W$ is the mean prefill waiting time, estimated using Kingman's approximation as described in Section~\ref{sub:der_opbatch}.

Setting the prefill occupancy in~\eqref{eq:operating_batch} to zero gives the \emph{full-pool batch} $B^{\max}$, the operating batch when the KV-cache pool is fully available to active decode requests:
\begin{equation}
B^{\max}=\frac{C_\mathrm{tok}-\mathbb{E}_r[(\ell_\mathrm{in}+R)^2]/\big(2\,\mathbb{E}_r[\ell_\mathrm{in}+R]\big)}{\bar{\ell}_{\mathrm{ctx}}+R}.
\label{eq:bmax}
\end{equation}
With prefill reservations, the operating batch is smaller than $B^{\max}$.

Substituting the expressions for $\bar{t}_D$ from~\eqref{eq:td_mean} and $\mu_D$ from~\eqref{eq:mud} into the memory balance~\eqref{eq:operating_batch} together with the waiting-time approximation yields a cubic equation in $B$. We show that this equation has a unique solution within the feasible operating regime, with details given in Theorem~\ref{thm:unique_b}.

\subsection{Power of a Provisioned Deployment}
\label{sub:power}

We next model the average power consumed by a provisioned deployment over a long serving window.
The power consumption of an inference instance is closely related to the utilization of its underlying hardware resources.
We find that a simple and effective characterization is the serving throughput of an instance relative to its serving capacity.
Intuitively, an instance operating close to its serving capacity keeps the bottleneck hardware resource highly utilized, whereas an instance serving only a small fraction of its serving capacity leaves more of that resource idle.
We therefore model power from the ratio between the instance serving throughput and its capacity.
Let $\lambda_{P,i}$ denote the serving throughput of prefill instance $i$, and $\lambda_{D,i}$ that of decode instance $i$.
We define their normalized serving throughputs as
\begin{equation}
\tilde\lambda_{P,i} = \frac{\lambda_{P,i}}{\mu_P},
\qquad
\tilde\lambda_{D,i} = \frac{\lambda_{D,i}}{\mu_D^{\max}},
\label{eq:lamn}
\end{equation}
where $\mu_D^{\max}$ denotes the \emph{full-pool decode capacity}, defined as the serving capacity when the KV-cache pool is fully available to active decode requests, i.e., in the absence of prefill reservations. Equivalently, $\mu_D^{\max}=\mu_D(B^{\max})$, the decode capacity of~\eqref{eq:mud} evaluated at the full-pool batch of~\eqref{eq:bmax}. Requests in prefill reserve part of the KV-cache pool, so the operating batch of~\eqref{eq:operating_batch} is smaller than $B^{\max}$.
A decode instance may therefore consume less power than at saturation even when the provisioned deployment operates at its serving capacity.
A value of $\tilde\lambda=0$ corresponds to an idle instance, while $\tilde\lambda=1$ corresponds to its reference capacity.

Across both prefill and decode and all calibration workloads in Appendix~\ref{app:calibration}, per-instance power increases approximately linearly with normalized serving throughput at low to moderate load and then saturates at high load, as shown in Figure~\ref{fig:ramp}.
We therefore model the per-instance power $p(\tilde\lambda)$ using a capped linear function, with separate parameters for the prefill and decode instances:

\begin{figure}[tb]
  \centering
  \includegraphics[width=\columnwidth]{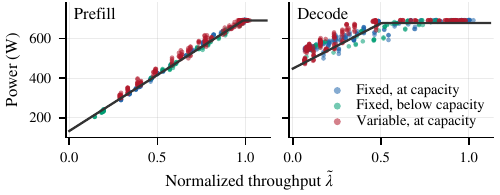}
  \caption{Per-GPU power versus normalized serving throughput for prefill and decode.
    Points show measurements, and curves show the fitted capped-ramp models.}
  \label{fig:ramp}
\end{figure}

\begin{equation}
p_s(\tilde\lambda_s)
=
\min\bigl(
p_{0,s}+\gamma_s\tilde\lambda_s,\;
p_{\mathrm{sat},s}
\bigr),
\qquad s\in\{P,D\},
\label{eq:ramp}
\end{equation}
where $s\in\{P,D\}$ indexes the prefill and decode instances,
$p_{0,s}$ is the static power floor,
$\gamma_s$ is the slope before saturation,
and $p_{\mathrm{sat},s}$ is the saturated power.

The deployment power $\mathcal{P}$ is the sum of the power consumed by its prefill and decode instances:
\begin{equation}
\mathcal{P} = \sum_{i=1}^{n_P} p_P(\tilde\lambda_{P,i}) \;+\; \sum_{i=1}^{n_D} p_D(\tilde\lambda_{D,i}) .
\label{eq:wtot}
\end{equation}
The sum in~\eqref{eq:wtot} holds for arbitrary per-instance serving throughputs.
We assume that the load is balanced, so that the request rate $\lambda$ is divided evenly among the instances of each pool, $\lambda_{P,i}=\lambda/n_P$ and $\lambda_{D,i}=\lambda/n_D$, and the deployment power reduces to $\mathcal{P}=n_P\,p_P\bigl(\tfrac{\lambda}{n_P\mu_P}\bigr)+n_D\,p_D\bigl(\tfrac{\lambda}{n_D\mu_D^{\max}}\bigr)$.
When comparing provisioned deployments, we evaluate this power at the serving capacity, $\lambda=\mu$, so that it depends only on $(n_P,n_D)$ and the workload $\mathcal{W}$.

\section{Serving Capacity Model for Decode Instances}
\label{sec:derivation}
This section derives the serving-capacity model for decode instances introduced in Section~\ref{sec:model}.
Throughout the analysis, $\mathbb{E}_r$, $\mathbb{E}_k$, and $\mathbb{E}_t$ denote averages over admitted requests, decode iterations, and continuous time, respectively.
The analysis applies to arrival processes and workloads for which these averages converge to the corresponding expectations as the observation window grows.

\subsection{Decode Capacity Model}
\label{sub:der_mud}

Unlike prefill, decode processes a changing set of active requests across iterations. This subsection provides a detailed derivation of the results presented in Section~\ref{sub:mud}.
Let $\mathcal{D}_k$ denote the set of requests active in iteration $k$, and let $B_k=|\mathcal{D}_k|$ be the instantaneous batch size.
We therefore define the operating batch as $B:=\mathbb{E}_k[B_k]$.

Each decode iteration reads the model weights once, requiring $2N$ bytes in bf16, and accesses the KV cache of every active request, transferring $\kappa$ bytes per context token.
For request $r\in\mathcal{D}_k$, let $\ell_{\mathrm{ctx},r,k}$ denote its context length in iteration $k$, consisting of its input tokens and the output tokens generated before that iteration.
The memory traffic in iteration $k$ is therefore
\begin{equation}
Q_k = \underbrace{2N}_{\text{model weights}}
    + \underbrace{\kappa\sum_{r\in\mathcal{D}_k}\ell_{\mathrm{ctx},r,k}}_{\text{KV cache}} .
\label{eq:qk}
\end{equation}
Under the roofline model, dividing $Q_k$ by the achieved memory bandwidth gives the decode-iteration time
$t_{D,k}=\frac{Q_k}{\beta\,\mathrm{MBU}}$.
Averaging $t_{D,k}$ over decode iterations gives the mean iteration time
\begin{equation}
\bar{t}_D := \mathbb{E}_k[t_{D,k}]
= \frac{2N}{\beta\,\mathrm{MBU}} + \frac{\kappa B\bar{\ell}_{\mathrm{ctx}}}{\beta\,\mathrm{MBU}} .
\label{eq:etd}
\end{equation}
By definition, $\bar{\ell}_{\mathrm{ctx}}$ is the average over active request--iteration pairs, so the mean aggregate context length is $B\bar{\ell}_{\mathrm{ctx}}$.
Importantly, $\bar{\ell}_{\mathrm{ctx}}$ is an active-request average rather than an average over arriving requests.
Requests with longer outputs remain active for more decode iterations and therefore contribute more often to the aggregate context length.
The following lemma expresses $\bar{\ell}_{\mathrm{ctx}}$ in terms of the workload distribution.

\begin{lemma}[Mean active context length]
\label{lem:ctx}
Suppose each request has an input--output length pair $(\ell_{\mathrm{in}},\ell_{\mathrm{out}})$ that follows the workload distribution $\mathcal{L}$, independently across requests.
Then the mean context length over active request--iteration pairs satisfies
\begin{equation}
\begin{aligned}
\bar{\ell}_{\mathrm{ctx}}
&:= \frac{\mathbb{E}_k\!\left[\sum_{r\in\mathcal{D}_k}\ell_{\mathrm{ctx},r,k}\right]}{\mathbb{E}_k[B_k]} \\
&= \underbrace{\mathbb{E}_r[\ell_\mathrm{in}]+\frac{\mathbb{E}_r[\ell_\mathrm{out}]}{2}}_{\text{arrival mean}}
+ \underbrace{\frac{\frac{1}{2}\mathrm{Var}(\ell_\mathrm{out})+\mathrm{Cov}(\ell_\mathrm{in},\ell_\mathrm{out})}{\mathbb{E}_r[\ell_\mathrm{out}-1]}}_{\text{residence correction}} .
\end{aligned}
\label{eq:ctx_longrun}
\end{equation}
\end{lemma}
The variance term arises because requests with longer outputs appear in more active request--iteration pairs.
The covariance term raises $\bar{\ell}_{\mathrm{ctx}}$ when longer outputs tend to occur with longer inputs and lowers it when they tend to occur with shorter inputs~\cite{gallager1996discrete,afdprovisioning2026}.
The detailed proof is given in Appendix~\ref{app:proof_ctx}.

The mean iteration time above describes the cost of one decode iteration.
We next convert it into request-level serving capacity.
Each decode iteration generates one token for every active request, or $B$ decode tokens on average.
Since each request requires $\mathbb{E}_r[\ell_\mathrm{out}-1]$ decode tokens on average, the decode instance executes $\mathbb{E}_r[\ell_\mathrm{out}-1]/B$ iterations per admitted request on average.
Multiplying by the mean iteration time $\bar t_D$ gives
\[
\underbrace{\frac{1}{\mu_D}}_{\text{time per request}}
=
\underbrace{\frac{\mathbb{E}_r[\ell_\mathrm{out}-1]}{B}}_{\text{iterations per request}}
\cdot
\underbrace{\vphantom{\frac{1}{B}}\bar t_D}_{\text{time per iteration}},
\]
which is the decode-capacity relation in~\eqref{eq:mud}.
The detailed derivation in Appendix~\ref{app:proof_mud} accounts explicitly for the time-varying instantaneous batch size $B_k$.

Because this relation depends on decode execution through the mean iteration time $\bar t_D$, it also applies when the calibrated execution-time overheads in Appendix~\ref{sub:calibration} are included.
The remaining quantity is the operating batch $B$, which we determine from the KV-cache memory balance next.

\subsection{Operating Batch of a Decode Instance}
\label{sub:der_opbatch}
 This subsection provides a detailed derivation of the results presented in Section~\ref{sub:operating_batch}.
We now derive the KV-cache memory balance in~\eqref{eq:operating_batch}, which determines the operating batch $B$.
After accounting for model parameters and runtime workspace, a decode instance has $C_\mathrm{tok}$ token slots available for its KV-cache pool, each with space for the key and value states of one token.
In each decode iteration $k$, these slots are divided into three components.
Let $\mathcal{P}_k$ denote the admitted requests whose prefill has not completed; each request $r\in\mathcal{P}_k$ holds a reservation of $\ell_{\mathrm{in},r}+R$ token slots.
Each request $r\in\mathcal{D}_k$ active in decode occupies $\ell_{\mathrm{ctx},r,k}$ slots for its current context and reserves $R$ further slots.
The remaining slots, denoted $U_k$, are unused.
At the start of iteration $k$ the three components fill the pool exactly, $C_\mathrm{tok}=\sum_{r\in\mathcal{P}_k}(\ell_{\mathrm{in},r}+R)+\sum_{r\in\mathcal{D}_k}(\ell_{\mathrm{ctx},r,k}+R)+U_k$.
Averaging the number of slots in each component over decode iterations gives the mean KV-cache memory balance:
\begin{equation}
\begin{aligned}
C_\mathrm{tok}
&=
\underbrace{
  \mathbb{E}_k\!\left[
    \sum_{r\in\mathcal{P}_k}(\ell_{\mathrm{in},r}+R)
  \right]
}_{\text{prefill occupancy}}
\\[4pt]
&\quad+
\underbrace{
  \mathbb{E}_k\!\left[
    \sum_{r\in\mathcal{D}_k}(\ell_{\mathrm{ctx},r,k}+R)
  \right]
}_{\text{decode occupancy}}
+
\underbrace{
  \mathbb{E}_k[U_k]
}_{\text{unused capacity}}.
\end{aligned}
\label{eq:mean_memory_balance}
\end{equation}
We next derive the three terms in~\eqref{eq:mean_memory_balance}: the mean prefill occupancy, the mean decode occupancy, and the mean unused capacity.

\paragraph{Prefill occupancy.}
A request is admitted once KV-cache space has been reserved on its assigned decode instance.
Its prefill waiting time $t_W$ is the time from admission to the start of prefill.
Before prefill completes, each admitted request holds a KV-cache reservation on its assigned decode instance.
Its contribution to prefill occupancy therefore depends on both the reservation size and its prefill holding time, the time spent waiting for and undergoing prefill.
At the decode serving limit, each instance admits requests at its serving capacity.
The following lemma relates the continuous-time mean occupancy to these request-level quantities.

\begin{lemma}[Mean prefill occupancy]
\label{lem:prefill}
Let $\mathcal{P}(t)$ denote the admitted requests whose prefill has not completed at time $t$.
Each request holds $\ell_\mathrm{in}+R$ token slots while waiting for and undergoing prefill.
Then the continuous-time mean occupancy of these reservations is
\begin{equation}
\mathbb{E}_t\!\Big[\sum_{r\in\mathcal{P}(t)}(\ell_{\mathrm{in},r}+R)\Big]
= \mu_D\,\mathbb{E}_r\!\left[(t_W+t_P)(\ell_\mathrm{in}+R)\right] .
\label{eq:prefill_ct}
\end{equation}
If, in addition, prefill serves requests in first-come, first-served (FCFS) order, request lengths are independent across arrivals, and $t_P=a_P\ell_\mathrm{in}+b_P\ell_\mathrm{in}^2$, then
\begin{equation}
\begin{aligned}
\mathbb{E}_r\!\left[(t_W+t_P)(\ell_\mathrm{in}+R)\right]
={}& \bar{t}_W\big(\mathbb{E}_r[\ell_\mathrm{in}]+R\big) \\
&+ a_P\big(\mathbb{E}_r[\ell_\mathrm{in}^2]+R\,\mathbb{E}_r[\ell_\mathrm{in}]\big) \\
&+ b_P\big(\mathbb{E}_r[\ell_\mathrm{in}^3]+R\,\mathbb{E}_r[\ell_\mathrm{in}^2]\big),
\end{aligned}
\label{eq:prefill_terms}
\end{equation}
where $\bar{t}_W=\mathbb{E}_r[t_W]$ is the mean prefill waiting time.
\end{lemma}
The occupancy relation~\eqref{eq:prefill_ct} is Little's law applied to KV-cache reservations, with reservation size as the weight~\cite{kleinrock1975,gallager1996discrete}.
The factorization in~\eqref{eq:prefill_terms} holds because, under FCFS and independent request lengths, a request's waiting time is determined by the work already in the queue and is therefore independent of its own input length.
The detailed proof is given in Appendix~\ref{app:proof_prefill}.

We estimate the mean prefill waiting time $\bar{t}_W$ using Kingman's approximation~\cite{kingman1961}:
\begin{equation}
\bar{t}_W \approx
\frac{CV_a^2+CV_s^2}{2}\cdot\frac{\rho_P}{\mu_P\,(1-\rho_P)},
\label{eq:kingman}
\end{equation}
with $\rho_P=n_D\mu_D/(n_P\mu_P)$ the serving utilization of a prefill instance, and $CV_a^2$ and $CV_s^2=\operatorname{Var}(t_P)/\mathbb{E}_r[t_P]^2$ the variation in its arrival intervals and prefill service times.
The approximation allows general distributions of service times and interarrival intervals under i.i.d. assumptions.
Here, prefill service time depends on input length, so we compute $CV_s^2$ from the workload's input-length distribution.
For Poisson arrivals assigned to prefill instances in round-robin order, $CV_a^2=1/n_P$.

Lemma~\ref{lem:prefill} gives a continuous-time average, whereas the KV-cache memory balance in~\eqref{eq:mean_memory_balance} averages occupancy over decode iterations.
We approximate the latter by the former:
\begin{equation}
\begin{aligned}
\mathbb{E}_k\!\left[
  \sum_{r\in\mathcal{P}_k}(\ell_{\mathrm{in},r}+R)
\right]
&\approx
\mathbb{E}_t\!\left[
  \sum_{r\in\mathcal{P}(t)}(\ell_{\mathrm{in},r}+R)
\right].
\end{aligned}
\label{eq:prefill_occupancy}
\end{equation}
This approximation is expected to be accurate when the decode-iteration duration is only weakly correlated with the prefill occupancy.
The evaluation in Section~\ref{sec:exp} uses the complete model with this approximation.
Substituting the waiting-time estimate in~\eqref{eq:kingman} into Lemma~\ref{lem:prefill} and applying the sampling approximation in~\eqref{eq:prefill_occupancy} gives the prefill-occupancy term in the memory balance.

\paragraph{Decode occupancy.}
Each active request occupies KV-cache space for its current context and reserves an additional $R$ token slots.
Averaging both parts over decode iterations gives the mean decode occupancy
\begin{equation}
\mathbb{E}_k\!\left[\sum_{r\in\mathcal{D}_k}(\ell_{\mathrm{ctx},r,k}+R)\right]
= B\left(\bar{\ell}_{\mathrm{ctx}}+R\right),
\label{eq:decode_occupancy}
\end{equation}
since, by the definitions of $B$ and $\bar{\ell}_{\mathrm{ctx}}$, the mean total context length is $B\bar{\ell}_{\mathrm{ctx}}$ and the $B_k$ reservations of $R$ slots average to $BR$.

\paragraph{Unused capacity.}
Requests reserve discrete amounts of KV-cache space, so the pool cannot always be filled exactly, and some space remains unused when the next reservation does not fit.
We approximate the unused space by assuming that the reservation that does not fit in the remaining space is sampled in proportion to its size and that the KV-cache limit falls uniformly within it.
The mean unused space is then
\begin{equation}
\mathbb{E}_k[U_k]\approx
\frac{\mathbb{E}_r[(\ell_{\mathrm{in}}+R)^2]}
     {2\,\mathbb{E}_r[\ell_{\mathrm{in}}+R]}.
\label{eq:unused}
\end{equation}
Appendix~\ref{app:unused} gives the derivation.

\paragraph{Operating batch.}
Combining the three occupancy terms gives the memory balance in~\eqref{eq:operating_batch}.
Let $g(B)$ denote the difference between the modeled KV-cache occupancy and the available pool capacity for a batch size $B$:
\[
\begin{aligned}
g(B)={}&
\mu_D(B)\left(
\bar t_W(B)\,\mathbb{E}_r[\ell_{\mathrm{in}}+R]
+\mathbb{E}_r[t_P(\ell_{\mathrm{in}}+R)]
\right)\\
&+B(\bar\ell_{\mathrm{ctx}}+R)
+\mathbb{E}_k[U_k]-C_{\mathrm{tok}}.
\end{aligned}
\]
The operating batch satisfies $g(B)=0$, where the modeled occupancy matches the available capacity.
The following theorem establishes its existence and uniqueness within the stable operating range.
\begin{theorem}[Existence and uniqueness of the operating batch]
\label{thm:unique_b}
If $C_{\mathrm{tok}}>\mathbb{E}_k[U_k]$,
then there is a unique operating batch $B>0$ satisfying
$g(B)=0$ and $\rho_P(B)<1$.
\end{theorem}

The theorem guarantees a unique operating batch within the stable prefill regime; the proof is given in Appendix~\ref{app:proof_unique_b}.
The condition $C_{\mathrm{tok}}>\mathbb{E}_k[U_k]$ is readily satisfied in the deployments considered here: the KV-cache pool accommodates multiple concurrent requests and is substantially larger than the mean unused space.
The solution lies below both the full-pool batch $B^{\max}$ and $B^\rho$, the upper boundary of the range $\rho_P(B)<1$.
We compute $B$ by solving $g(B)=0$ numerically within these bounds.
Substituting the resulting batch into~\eqref{eq:mud} gives the decode capacity, which determines the deployment serving capacity through~\eqref{eq:twobottleneck}.

\section{Experiments}
\label{sec:exp}
We first evaluate the accuracy of the analytical serving-capacity and power models against measurements
across PD-disaggregated deployments,
using one fixed-length workload and inference request samples from the Mooncake~\cite{mooncake2025} and Azure~\cite{splitwise2024} production traces.
For serving capacity, we test generalization from fixed-length calibration workloads
to variable-length workloads derived from production traces;
for power, we test cross-workload and cross-deployment accuracy.
We then compare the modeled and measured serving capacity--power Pareto fronts
and evaluate whether they select the same minimum-power feasible deployment for a given serving-capacity requirement.
Finally, we demonstrate two provisioning scenarios:
reducing provisioned resources after the request rate falls,
and temporarily reducing power under an acceptable serving-capacity threshold to respond to grid-side load-reduction signals.

\begin{table}[tb]
\centering
\caption{Request counts and input and output lengths for the experimental workloads.
  The fixed-length column lists the lengths used across the fixed-length workloads; the trace columns report sample means.}
\label{tab:workload}
\footnotesize\setlength{\tabcolsep}{4pt}
\begin{tabular}{lrrr}
\toprule
 & Fixed-length & Mooncake & Azure \\
\midrule
\# requests                       & ---            & $1000$ & $7200$ \\
$\mathbb{E}_r[\ell_\mathrm{in}]$  & $1024$--$8192$ & $8307$ & $1158$ \\
$\mathbb{E}_r[\ell_\mathrm{out}]$ & $256$          & $345$  & $211$  \\
\bottomrule
\end{tabular}
\end{table}
\begin{table}[tb]
\centering
\caption{Calibrated model parameters}
\label{tab:calib}
\small
\begin{tabular}{llll}
\toprule
 & Symbol & Prefill & Decode \\
\midrule
\multicolumn{4}{l}{\emph{Serving capacity}}\\
Hardware utilization & MFU, MBU        & $0.67$   & $0.77$ \\
Attention coefficient& $c_a$           & $2.17$   & \\
Per-iteration overhead & $t_\mathrm{iter}$ &      & $1.00$ ms \\
Per-request overhead & $t_\mathrm{req}$ &    & $0.062$ ms \\
\midrule
\multicolumn{4}{l}{\emph{Power ramp} (W)}\\
Static power         & $p_0$           & $133$    & $448$ \\
Slope                & $\gamma$        & $566$    & $458$ \\
Saturated power      & $p_\mathrm{sat}$& $692$    & $678$ \\
\bottomrule
\end{tabular}
\end{table}
\subsection{Experimental Setup}
\label{sub:exp_setup}
All experiments serve Qwen3-32B~\cite{qwen3_2025} on H200 GPUs~\cite{h200_2023}, with one GPU per instance.
We use SGLang~0.5.9~\cite{sglang2024}, an open-source LLM serving framework.
For KV-cache transfers between prefill and decode instances, we use Mooncake~\cite{mooncake2025} as the backend.
Within the prefill and decode pools, requests are assigned to instances in round-robin order, so successive requests are distributed across the instances in turn.

We conduct experiments with three types of workloads.
(i) Fixed-length workloads, in which all requests have the same input and output lengths, using five discrete input lengths between 1024 and 8192 tokens and a fixed output length of 256 tokens.
(ii) Request samples from the Mooncake conversation trace~\cite{mooncake2025}, collected from Moonshot AI's Kimi chat service and containing long-context requests.
(iii) Request samples from the Azure conversation trace~\cite{splitwise2024}, collected from a production LLM inference service on Microsoft Azure.
Both traces record the input and output token counts of individual requests.
Table~\ref{tab:workload} summarizes the requests used in our experiments.
Section~\ref{sub:exp_e2e} evaluates one fixed-length workload and both trace-derived workloads
across deployments with 2 to 8 instances.
Each evaluated deployment operates at its serving capacity.

Table~\ref{tab:calib} lists the calibrated model parameters used in the evaluation.
For all serving-capacity and power comparisons, we exclude startup and shutdown transients from the measurement windows.
Reported total power includes only GPUs assigned to instances currently in service.
Appendix~\ref{app:calibration} provides the model and hardware constants and describes the calibration, measurement, and trace-sampling procedures.

\begin{figure}[!tb]
  \centering
  \includegraphics[width=\columnwidth]{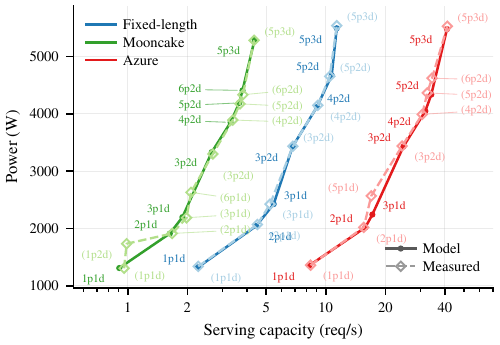}
  \caption{Modeled and measured serving capacity--power Pareto fronts for the three workloads.}
  \label{fig:pareto_traces}
\end{figure}

\subsection{Model Validation}
\label{sub:exp_e2e}

\paragraph{Serving-capacity accuracy.}
We evaluate whether the analytical serving-capacity model, calibrated only on fixed-length workloads,
remains accurate for variable-length workloads derived from production traces without further calibration.
We first evaluate the fixed-length workload with 4096 input and 256 output tokens,
which is included in the calibration workloads.
Figure~\ref{fig:pareto_traces} compares the modeled and measured serving capacities for this workload.
The mean absolute percentage error across all evaluated deployments is $1.2\%$.

We then evaluate the Mooncake and Azure trace-derived workloads,
neither of which is used to calibrate the serving-capacity model.
We keep the calibrated model parameters fixed and update only the workload length moments for each workload.
The mean absolute percentage errors are $3.0\%$ for Mooncake and $1.6\%$ for Azure.

\paragraph{Power accuracy.}
We evaluate whether the same prefill and decode power models remain accurate
across the evaluated workloads and deployments.
Each model takes normalized serving throughput as input, and we keep its calibrated parameters fixed throughout the evaluation.

The comparisons are visualized in Figure~\ref{fig:pareto_traces}. For the fixed-length workload, the mean absolute percentage error in total GPU power across all evaluated deployments is $2.3\%$.
The corresponding errors are $2.7\%$ for Mooncake and $2.6\%$ for Azure.
These results show that the calibrated power models closely match measured deployment power across the evaluated workloads.

\paragraph{Pareto fronts and deployment selection.}
Figure~\ref{fig:pareto_traces} shows that the modeled serving capacity--power Pareto fronts
closely follow the measured fronts.
Along the measured fronts, higher serving capacity comes with higher total GPU power.
At comparable power levels, serving capacity is highest for Azure, followed by the fixed-length workload and Mooncake.
This ordering is consistent with the mean input and output lengths in Table~\ref{tab:workload},
since longer requests require more processing per request.

We then evaluate whether the modeled and measured Pareto fronts lead to the same deployment choice
for a given serving-capacity requirement.
For each evaluated requirement, we select the minimum-power feasible deployment
separately from the modeled and measured fronts.
The selected deployments agree for $90\%$ of the evaluated serving-capacity requirements for the fixed-length workload,
$85\%$ for Mooncake, and $86\%$ for Azure.
All disagreements occur near a feasibility boundary, where the required serving capacity is within $4\%$
of a deployment's measured serving capacity.

\subsection{Power-Aware Provisioning}
\label{sub:exp_track}
We use two scenarios to examine how provisioning changes can improve efficiency
or provide temporary power flexibility.
In the first, a decrease in request rate creates serving-capacity headroom
that allows provisioned resources to be reduced while the deployment remains stable.
In the second, the request rate remains unchanged and provisioned capacity is temporarily reduced below that rate, reducing power at the cost of transient service degradation.
Both scenarios use the fixed-length workload with 4096 input and 256 output tokens
and begin with three prefill instances and two decode instances (3p2d).

To capture the transient effects of reconfiguration, we report TTFT and TPOT throughout each run.
TTFT is measured from request arrival to the first output token and includes prefill-queue waiting time, prefill execution, and KV-cache transfer to the decode instance.
TPOT is measured over individual intervals between consecutive output tokens rather than first averaging within each request, so that its p95 captures transient stalls.

\subsubsection{Efficiency After a Request-Rate Drop}
We examine whether a decrease in request rate creates sufficient capacity headroom to remove a decode instance.
The request rate is initially set to $60\%$ of the serving capacity of 3p2d and is then reduced to $3.1$\,req/s (Figure~\ref{fig:scn_pair}, left).
Before any change in provisioning, the lower request rate alone reduces total power by approximately $270$\,W, while TTFT and TPOT remain nearly unchanged.
This reduction reflects the load dependence of per-instance power: the same provisioned deployment consumes less power at a lower request rate.

At $3.1$\,req/s, 3p1d still provides sufficient serving capacity, so one decode instance can be removed without overloading the resulting deployment.
This reconfiguration reduces total power by an additional $560$\,W.
The resulting 3p1d deployment operates at $60\%$ of its serving capacity.
Median TTFT remains approximately $0.55$\,s, while median TPOT increases from $23$ to $28$\,ms as the remaining decode instance operates at higher load.
Thus, after the request-rate reduction, capacity headroom can be converted into additional power savings without inducing queue growth; the main latency effect is a modest increase in decode-side TPOT.

\begin{figure}[!htb]
  \centering
  \includegraphics[width=\columnwidth]{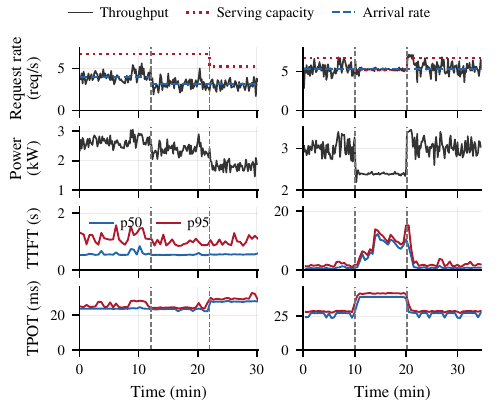}
  \caption{Serving throughput, total GPU power, TTFT, and TPOT during decode-instance reconfiguration.
    Left: a decode instance is removed after the request rate decreases.
    Right: a decode instance is removed and restored at a fixed request rate.}
  \label{fig:scn_pair}
\end{figure}

\subsubsection{Power Flexibility at a Fixed Request Rate}
We examine how much power can be saved by temporarily reducing provisioned capacity below a fixed request rate, and how long the reduction can be sustained.
We hold the request rate at $5.3$\,req/s, or $80\%$ of the serving capacity of 3p2d.
We then remove one decode instance for $10$\,min before restoring it (Figure~\ref{fig:scn_pair}, right).

Removing the decode instance reduces total power by approximately $620$\,W but lowers serving capacity to $5.2$\,req/s, slightly below the offered request rate.
With serving capacity below the request rate, requests accumulate in the queue and median TTFT increases from $0.7$ to $8.5$\,s during the removal interval.
Meanwhile, the remaining decode instance operates near its serving limit with a larger batch, increasing median TPOT from $27$ to $39$\,ms.
The much larger TTFT increase arises because requests must wait in the growing queue once the request rate exceeds serving capacity.

Restoring the second decode instance raises serving capacity above the request rate and allows the accumulated queue to drain within approximately $1.3$\,min.
TTFT, TPOT, and total power then return to their pre-removal levels.
Thus, when removing resources would reduce serving capacity below the request rate, the deployment can still provide temporary power flexibility, but the reduction can be sustained only as long as the resulting queueing delay remains tolerable.
Appendix~\ref{app:shrink} extends this behavior to a sequential scale-down experiment.

\section{Conclusion}
\label{sec:conclusion}

This paper develops an analytical framework for power-aware provisioning of PD-disaggregated AI inference. The framework characterizes how the workload distribution, hardware limits, and numbers of prefill and decode instances jointly determine serving capacity and average deployment power. In particular, the serving-capacity model captures the coupling between the two serving phases induced by KV-cache reservations, as well as the impact of request queueing, while the power model relates per-instance power consumption to normalized serving throughput. Combining these models yields a discrete serving capacity--power Pareto front over candidate provisioned deployments, providing a direct way to identify minimum-power deployments for a required serving capacity and to quantify the serving capacity available under different power limits.

Experiments across fixed-length workloads and workloads derived from production traces show that the analytical models closely reproduce measured serving capacity and power while reusing the same calibrated parameters across workloads. The resulting modeled Pareto fronts also lead to provisioning decisions that largely agree with those obtained from measurements. These results show that analytical characterization of the serving capacity--power relationship can reduce reliance on repeated profiling and simulation while providing a principled basis for efficient and grid-responsive operation of large-scale AI inference systems.

\bibliographystyle{ACM-Reference-Format}
\bibliography{references}

\appendix
\section{Additional Derivations and Proofs}
\label{app:proofs}
\subsection{Proof of Lemma~\ref{lem:ctx}}
\label{app:proof_ctx}
\begin{proof}
Let $A_K$ denote the number of requests admitted to a given decode instance during the first $K$ decode iterations.
Reindexing the context contributions from iterations to requests gives
\begin{equation}
\begin{aligned}
\bar{\ell}_{\mathrm{ctx}}
  &:= \frac{\mathbb{E}_k\!\left[\textstyle\sum_{r\in\mathcal{D}_k}\ell_{\mathrm{ctx},r,k}\right]}
          {\mathbb{E}_k[B_k]} \\
  &= \lim_{K\to\infty}
     \frac{\textstyle\sum_{k=1}^{K}\sum_{r\in\mathcal{D}_k}\ell_{\mathrm{ctx},r,k}}
          {\textstyle\sum_{k=1}^{K}B_k}\\
  &\stackrel{\circled{1}}{=}\!\! \lim_{K\to\infty}
     \frac{\textstyle\sum_{r=1}^{A_K}\sum_{j=1}^{\ell_{\mathrm{out},r}-1}(\ell_{\mathrm{in},r}+j)}
          {\textstyle\sum_{r=1}^{A_K}(\ell_{\mathrm{out},r}-1)} \\
  &\stackrel{\circled{2}}{=} \frac{\mathbb{E}_r\!\left[\textstyle\sum_{j=1}^{\ell_\mathrm{out}-1}(\ell_\mathrm{in}+j)\right]}
          {\mathbb{E}_r[\ell_\mathrm{out}-1]} \\
  &\stackrel{\circled{3}}{=} \frac{\mathbb{E}_r\!\left[(\ell_\mathrm{out}-1)\left(\ell_\mathrm{in}+\frac{\ell_\mathrm{out}}{2}\right)\right]}
          {\mathbb{E}_r[\ell_\mathrm{out}-1]} \\
  &= \frac{\mathbb{E}_r[\ell_\mathrm{out}-1]\,
           \mathbb{E}_r\!\left[\ell_\mathrm{in}+\frac{\ell_\mathrm{out}}{2}\right]
           + \mathrm{Cov}\!\left(\ell_\mathrm{out}-1,\,\ell_\mathrm{in}+\frac{\ell_\mathrm{out}}{2}\right)}
          {\mathbb{E}_r[\ell_\mathrm{out}-1]} \\
  &\stackrel{\circled{4}}{=} \underbrace{\mathbb{E}_r[\ell_\mathrm{in}]+\frac{\mathbb{E}_r[\ell_\mathrm{out}]}{2}}_{\text{arrival mean}}
   + \underbrace{\frac{\frac{1}{2}\mathrm{Var}(\ell_\mathrm{out})+\mathrm{Cov}(\ell_\mathrm{in},\ell_\mathrm{out})}
                      {\mathbb{E}_r[\ell_\mathrm{out}-1]}}_{\text{residence correction}} .
\end{aligned}
\label{eq:ctx_proof}
\end{equation}
Equation $\circled{1}$ counts the same context contributions request by request rather than iteration by iteration.
A request contributes $\ell_\mathrm{out}-1$ terms, one per decode iteration, and its context length increases from $\ell_\mathrm{in}+1$ to $\ell_\mathrm{in}+\ell_\mathrm{out}-1$ over them.
In $\circled{2}$, boundary requests contribute a vanishing fraction of both sums as $K\to\infty$, so the request averages converge to expectations under the workload distribution~\cite{gallager1996discrete}.
Equation $\circled{3}$ sums this arithmetic sequence, yielding $(\ell_\mathrm{out}-1)(\ell_\mathrm{in}+\ell_\mathrm{out}/2)$.
Finally, expanding the covariance in $\circled{4}$ gives~\eqref{eq:ctx_longrun}, which completes the proof.
\end{proof}

\subsection{Derivation of the Decode Serving Capacity}
\label{app:proof_mud}
\begin{proof}
Let $A_K$ denote the number of requests admitted to a given decode instance during the first $K$ decode iterations.
Each request requires $\ell_\mathrm{out}-1$ decode tokens, while each decode iteration generates one token for every active request.
Over a long serving window, the mean decode time per admitted request is therefore
\begin{equation}
\begin{aligned}
\frac{1}{\mu_D}
&\stackrel{\circled{1}}{=}
\lim_{K\to\infty}
\frac{\displaystyle\sum_{k=1}^{K}t_{D,k}}{A_K}
\\[6pt]
&=
\lim_{K\to\infty}
\frac{1}{A_K}
\frac{
\displaystyle\sum_{k=1}^{K}
\left(
2N+
\kappa\sum_{r\in\mathcal D_k}
\ell_{\mathrm{ctx},r,k}
\right)
}{
\beta\,\mathrm{MBU}
}
\\[6pt]
&=
\frac{1}{\beta\,\mathrm{MBU}}
\lim_{K\to\infty}
\frac{
2NK+
\kappa\displaystyle\sum_{k=1}^{K}
\sum_{r\in\mathcal D_k}
\ell_{\mathrm{ctx},r,k}
}{
A_K
}
\\[6pt]
&\stackrel{\circled{2}}{=}
\frac{2N}{\beta\,\mathrm{MBU}}
\left(
\lim_{K\to\infty}
\frac{\sum_{k=1}^{K}B_k}{A_K}
\right)
\bigg/
\left(
\lim_{K\to\infty}
\frac{\sum_{k=1}^{K}B_k}{K}
\right)
\\[4pt]
&\quad+
\frac{\kappa}{\beta\,\mathrm{MBU}}
\lim_{K\to\infty}
\frac{
\displaystyle\sum_{r=1}^{A_K}
\sum_{j=1}^{\ell_{\mathrm{out},r}-1}
(\ell_{\mathrm{in},r}+j)
}{
A_K
}
\\[6pt]
&\stackrel{\circled{3}}{=}
\frac{1}{\beta\,\mathrm{MBU}}
\left[
\frac{
2N\,\mathbb E_r[\ell_{\mathrm{out}}-1]
}{B}
+
\kappa\,
\mathbb E_r\!\left[
\sum_{j=1}^{\ell_{\mathrm{out}}-1}
(\ell_{\mathrm{in}}+j)
\right]
\right]
\\[6pt]
&\stackrel{\circled{4}}{=}
\frac{\mathbb E_r[\ell_{\mathrm{out}}-1]}{B}
\left(
\frac{2N}{\beta\,\mathrm{MBU}}
+
\frac{
\kappa B\bar{\ell}_{\mathrm{ctx}}
}{
\beta\,\mathrm{MBU}
}
\right)
\\[6pt]
&\stackrel{\circled{5}}{=}
\frac{
\mathbb E_r[\ell_{\mathrm{out}}-1]\,
\bar t_D
}{B}.
\end{aligned}
\label{eq:mud_proof}
\end{equation}
Equation $\circled{1}$ is the mean decode time per admitted request over the observation window, in which boundary requests contribute a vanishing fraction as $K\to\infty$.
Equations $\circled{2}$ and $\circled{3}$ relate iteration-level and request-level counts of decode work.
The total number of active request--iteration pairs, $\sum_k B_k$, is also the total number of decode tokens generated.
Hence $\sum_k B_k/A_K$ converges to $\mathbb{E}_r[\ell_\mathrm{out}-1]$, while $\sum_k B_k/K$ converges to the operating batch $B$.
The KV-cache contribution is reindexed by request as in~\eqref{eq:ctx_proof}.
Substituting the definitions of $\bar{\ell}_{\mathrm{ctx}}$ and $\bar{t}_D$ in $\circled{4}$ and $\circled{5}$ then yields the capacity relation in~\eqref{eq:mud}.
This relation also holds with calibrated overheads, since the token count is unchanged and the overheads are included in the mean iteration time $\bar t_D$.
\end{proof}

\subsection{Proof of Lemma~\ref{lem:prefill}}
\label{app:proof_prefill}
\begin{proof}
To compute the mean occupancy, consider continuous time $t$, and let $A(T)$ denote the number of requests admitted to a given decode instance during $[0,T]$.
Reindexing the accumulated space--time occupancy over this interval from time to requests gives
\begin{equation}
\begin{aligned}
&\mathbb{E}_t\!\left[
  \sum_{r\in\mathcal{P}(t)}(\ell_{\mathrm{in},r}+R)
\right] \\
&=
\lim_{T\to\infty}\frac{1}{T}\int_0^T
  \sum_{r\in\mathcal{P}(t)}(\ell_{\mathrm{in},r}+R)\,dt \\
&\stackrel{\circled{1}}{=}
\lim_{T\to\infty}\frac{1}{T}
  \sum_{r=1}^{A(T)}(t_{W,r}+t_{P,r})(\ell_{\mathrm{in},r}+R) \\
&=
\lim_{T\to\infty}\frac{A(T)}{T}
\left[
  \frac{1}{A(T)}\sum_{r=1}^{A(T)}
  (t_{W,r}+t_{P,r})(\ell_{\mathrm{in},r}+R)
\right] \\
&\stackrel{\circled{2}}{=}
\mu_D\,\mathbb{E}_r\!\left[(t_W+t_P)(\ell_\mathrm{in}+R)\right].
\end{aligned}
\label{eq:prefill_ct_proof}
\end{equation}
Equation $\circled{1}$ counts the same space--time occupancy request by request.
Request $r$ holds $\ell_{\mathrm{in},r}+R$ token slots for its prefill holding time $t_{W,r}+t_{P,r}$, its waiting time plus its prefill service time.
Requests whose reservations overlap either end of the observation window contribute only boundary terms, which vanish as $T\to\infty$.
In $\circled{2}$, the admission rate $A(T)/T$ converges to the decode capacity $\mu_D$ under steady operation at the serving limit, while the request average converges to $\mathbb{E}_r[(t_W+t_P)(\ell_\mathrm{in}+R)]$.
This establishes~\eqref{eq:prefill_ct}.
We next expand the expectation on its right-hand side:
\begin{equation}
\begin{aligned}
& \mathbb{E}_r\!\left[(t_W+t_P)(\ell_\mathrm{in}+R)\right] \\
={}& \mathbb{E}_r\!\left[t_W(\ell_\mathrm{in}+R)\right]
  + \mathbb{E}_r\!\left[t_P(\ell_\mathrm{in}+R)\right] \\
\stackrel{\circled{1}}{=}{}& \bar{t}_W\,\mathbb{E}_r[\ell_\mathrm{in}+R]
  + \mathbb{E}_r\!\left[t_P(\ell_\mathrm{in}+R)\right] \\
={}& \bar{t}_W\,\mathbb{E}_r[\ell_\mathrm{in}+R]
  + \mathbb{E}_r\!\left[
   (a_P\ell_\mathrm{in}+b_P\ell_\mathrm{in}^2)(\ell_\mathrm{in}+R)
   \right] \\[2pt]
={}& \bar{t}_W\big(\mathbb{E}_r[\ell_\mathrm{in}]+R\big) \\
& \quad + a_P\big(
  \mathbb{E}_r[\ell_\mathrm{in}^2]+R\,\mathbb{E}_r[\ell_\mathrm{in}]
  \big) \\
& \quad + b_P\big(
  \mathbb{E}_r[\ell_\mathrm{in}^3]+R\,\mathbb{E}_r[\ell_\mathrm{in}^2]
  \big).
\end{aligned}
\label{eq:prefill_terms_proof}
\end{equation}
Under FCFS, an arriving request's waiting time is determined by the work already in the queue.
With request lengths independent across arrivals, it is therefore independent of that request's own input length.
Equation $\circled{1}$ uses this independence to factor $\mathbb{E}_r[t_W(\ell_\mathrm{in}+R)]$ as $\bar t_W\,\mathbb{E}_r[\ell_\mathrm{in}+R]$, with $\bar t_W=\mathbb{E}_r[t_W]$.
The resulting expression is~\eqref{eq:prefill_terms}, which completes the proof.
\end{proof}

\subsection{Derivation of the Unused-Capacity Approximation}
\label{app:unused}
Let $A_K$ denote the number of requests admitted to a given decode instance during the first $K$ decode iterations.
Under the two steps of the approximation, the mean unused space is
\begin{equation}
\begin{aligned}
\mathbb{E}_k[U_k]
&\stackrel{\circled{1}}{\approx}
\lim_{K\to\infty}
\frac{\displaystyle\sum_{r=1}^{A_K}\int_0^{\ell_{\mathrm{in},r}+R}u\,du}
     {\displaystyle\sum_{r=1}^{A_K}(\ell_{\mathrm{in},r}+R)} \\
&=
\lim_{K\to\infty}
\frac{\displaystyle\sum_{r=1}^{A_K}(\ell_{\mathrm{in},r}+R)^2/2}
     {\displaystyle\sum_{r=1}^{A_K}(\ell_{\mathrm{in},r}+R)} \\
&\stackrel{\circled{2}}{=}
\frac{\mathbb{E}_r[(\ell_{\mathrm{in}}+R)^2]}
     {2\,\mathbb{E}_r[\ell_{\mathrm{in}}+R]}.
\end{aligned}
\label{eq:unused_proof}
\end{equation}
Equation $\circled{1}$ applies the two steps: the size weighting selects the reservation that does not fit, and the uniform position leaves on average half of it unused, so the mean unused space is the size-weighted mean of $s/2$.
Equation $\circled{2}$ divides numerator and denominator by $A_K$ and takes $K\to\infty$, which yields the ratio of workload expectations.
The result is~\eqref{eq:unused}.

\subsection{Proof of Theorem~\ref{thm:unique_b}}
\label{app:proof_unique_b}
\begin{proof}
On the stable range, increasing $B$ increases the decode capacity and therefore the prefill utilization.
The waiting time in~\eqref{eq:kingman} increases with utilization, so both the prefill occupancy and the decode occupancy increase with $B$.
The unused-capacity approximation in~\eqref{eq:unused} is independent of $B$.
Consequently, $g(B)$ is continuous and strictly increasing on this range.
Explicitly,
\[
\begin{aligned}
g'(B)={}&\mu_D'(B)\big(\bar t_W\,\mathbb{E}_r[\ell_\mathrm{in}+R]+\mathbb{E}_r[t_P(\ell_\mathrm{in}+R)]\big)\\[2pt]
&+\mu_D\,\bar t_W'(B)\,\mathbb{E}_r[\ell_\mathrm{in}+R]
+\bar\ell_{\mathrm{ctx}}+R>0 .
\end{aligned}
\]

At $B=0$, both prefill and decode occupancies vanish, giving $g(0)=\mathbb{E}_k[U_k]-C_\mathrm{tok}<0$.
To find a positive upper value, we compare the full-pool batch with the queue-stability limit.
The full-pool batch $B^{\max}$ of~\eqref{eq:bmax} is positive by assumption.
Let $B^\rho=\sup\{B>0:\rho_P(B)<1\}$ denote the queue-stability limit, the upper end of the stable range.

If $B^{\max}<B^\rho$, the prefill queue remains stable
at $B^{\max}$.
Substituting~\eqref{eq:bmax} into the residual gives
\[
g(B^{\max})
=
\mu_D(B^{\max})
\left(
\bar t_W(B^{\max})\,\mathbb{E}_r[\ell_{\mathrm{in}}+R]
+\mathbb{E}_r[t_P(\ell_{\mathrm{in}}+R)]
\right)>0.
\]
If $B^\rho\le B^{\max}$, then $\rho_P(B)\to1$
as $B\to(B^\rho)^-$.
By~\eqref{eq:kingman}, $\bar t_W(B)\to+\infty$,
while the definition of $\rho_P$ gives
$\mu_D(B)\to n_P\mu_P/n_D>0$.
Since $\mathbb{E}_r[\ell_{\mathrm{in}}+R]>0$
and the remaining terms have finite limits,
\[
\lim_{B\to(B^\rho)^-}g(B)
=
\lim_{B\to(B^\rho)^-}
\mu_D(B)\bar t_W(B)\,
\mathbb{E}_r[\ell_{\mathrm{in}}+R]
=+\infty.
\]

In either case, continuity and $g(0)<0$ guarantee a root
in $\big(0,\min(B^{\max},B^\rho)\big)$.
Since $g$ is strictly increasing on the stable range,
this root is unique.
\end{proof}

\section{Model Calibration}
\label{app:calibration}
The analytical models combine quantities obtained from the workload, model architecture, and hardware specifications with parameters estimated from measurements.
We calibrate MFU and $c_a$ for prefill; MBU, $t_\mathrm{iter}$, and $t_\mathrm{req}$ for decode; and the power-ramp parameters $p_0$, $\gamma$, and $p_\mathrm{sat}$ for each role.
The fitted parameters are reused across provisioned deployments and workloads under the same model, hardware, and serving configuration.
We first summarize the model and hardware constants, then describe the calibration method and report the calibration runs and fitted results.
The parameter values used in the evaluation are listed in Table~\ref{tab:calib}.

\subsection{Model and Hardware Constants}
\label{sub:constants}
The H200 GPU~\cite{h200_2023} has a peak compute throughput of $\pi=989$ TFLOP/s and a memory bandwidth of $\beta=4.8$ TB/s.
Qwen3-32B~\cite{qwen3_2025} has $N=32.8$B parameters and $L=64$ layers, with $h_q=64$ query heads and $h_\mathrm{kv}=8$ KV heads of width $d_\mathrm{head}=128$ in grouped-query attention~\cite{gqa2023}.
Two quantities follow from the architecture: the attention width $d=h_q d_\mathrm{head}=8192$, and the KV-cache bytes per context token $\kappa=4Lh_\mathrm{kv}d_\mathrm{head}$, since for each token the KV cache stores one key and one value per KV head in every layer, using two bytes per value in bf16.

\subsection{Calibration Method}
\label{sub:calibration}

\paragraph{Prefill.}
Prefill calibration estimates MFU and $c_a$ from the dependence of the prefill service time on input length.
Because~\eqref{eq:tp} contains both a linear and a quadratic term, the fit requires $t_P$ at several input lengths.
For each input length, we saturate a single prefill instance while provisioning enough decode capacity to keep prefill as the bottleneck.
The measured completion rate is then $\mu_P$, and because each calibration workload uses a fixed input length,~\eqref{eq:mup} reduces to $t_P=1/\mu_P$.
We fit $a_P\ell_\mathrm{in}+b_P\ell_\mathrm{in}^2$ to the measured $1/\mu_P$, with no intercept as required by~\eqref{eq:tp}.
The fitted coefficients give
\[
\mathrm{MFU}=\frac{2N}{\pi a_P},
\qquad
c_a=\frac{\pi\,\mathrm{MFU}\,b_P}{Ld} .
\]

\paragraph{Decode.}
We calibrate the decode iteration time rather than decode capacity, because it also depends on the operating batch determined by the KV-cache memory balance of Section~\ref{sub:operating_batch}.
At a fixed context length, the memory-traffic model~\eqref{eq:td_mean} determines an iteration time that is linear in the batch, with intercept $a_D$ and slope $b_D$.
Measurements approximately follow this dependence, but also indicate an iteration-wide overhead and a context-independent cost per active request.
We account for these costs by adding $t_\mathrm{iter}$ to the intercept and $t_\mathrm{req}$ to the slope:
\begin{equation}
a_D = \frac{2N}{\beta\,\mathrm{MBU}} + t_\mathrm{iter},
\qquad
b_D = \frac{\kappa\bar{\ell}_{\mathrm{ctx}}}{\beta\,\mathrm{MBU}} + t_\mathrm{req} .
\label{eq:etd_cal}
\end{equation}
These overheads may include fixed iteration costs, such as CPU--GPU dispatch and kernel launches, and context-independent per-request costs, such as metadata processing.
Both are empirical corrections and need not correspond to any single mechanism.
Algorithm~\ref{alg:capacity} uses these calibrated coefficients to compute the operating batch and decode capacity.

To estimate these parameters, we first fit the dependence on batch size at each context-length setting, then fit the resulting slopes as a function of context length.
Let $j=1,\ldots,J$ index the context-length settings, and let $\bar\ell_{\mathrm{ctx},j}$ denote the context length of setting $j$.
The calibration workloads use fixed input and output lengths, so $\mathrm{Var}(\ell_\mathrm{out})=\mathrm{Cov}(\ell_\mathrm{in},\ell_\mathrm{out})=0$;
the residence correction in~\eqref{eq:ctx_longrun} therefore vanishes and $\bar\ell_{\mathrm{ctx},j}=\ell_{\mathrm{in},j}+\ell_\mathrm{out}/2$.
At each setting $j$, we fit $\bar t_{D,j}=a_{D,j}+b_{D,j}B$ to the measured mean iteration times, where $a_{D,j}$ and $b_{D,j}$ are the fitted intercept and slope.
Because the fitted intercepts vary only modestly across context lengths, we use their mean as the context-independent intercept assumed by the model.
We then regress $b_{D,j}$ on $\bar\ell_{\mathrm{ctx},j}$ as $b_{D,j}=c_{D,1}\bar\ell_{\mathrm{ctx},j}+c_{D,0}$, where $c_{D,0}$ and $c_{D,1}$ are the fitted intercept and slope.
The three decode parameters are then
\[
\begin{aligned}
\mathrm{MBU}
&=\frac{\kappa}{\beta c_{D,1}},\\
t_{\mathrm{req}}
&=c_{D,0},\\
t_{\mathrm{iter}}
&=\frac{1}{J}\sum_{j=1}^{J}a_{D,j}
-\frac{2N}{\beta\,\mathrm{MBU}} .
\end{aligned}
\]
These values are substituted into~\eqref{eq:etd_cal} in all reported results.

\paragraph{Power.}
We calibrate the power model by fitting the capped ramp of~\eqref{eq:ramp} to measured per-instance power at different values of normalized serving throughput.
The measurements must cover both the region where power rises with the serving throughput and the region where it has saturated.
Because the throughput at which power saturates is not known in advance, we fit $p_0$, $\gamma$, and $p_\mathrm{sat}$ jointly, using measurements from all calibration workloads to obtain one ramp for each role.
The shared ramp per role assumes that the normalized serving throughput accounts for the workload dependence of average power.

\subsection{Calibration Runs and Results}
\label{sub:calibration_results}
We use fixed-length workloads to calibrate the serving-capacity model.
Every request has 256 output tokens, while the input length takes one of five values: 1024, 1536, 2048, 4096, and 8192 tokens.
Prefill calibration uses the three longer input lengths, and decode calibration uses all five.
Power calibration additionally includes measurements from the request-rate sweep and from runs using workloads derived from the production traces.

\paragraph{Prefill.}
At each of the three longer input lengths, we run 5 provisioned deployments, with 1 prefill instance and 1 to 5 decode instances, at an arrival rate above their serving capacity, 15 runs in all.
We obtain the completion rate by dividing the decode token-generation rate over the measurement window by $\ell_\mathrm{out}-1$,
the number of tokens a request receives in decode, since prefill produces its first output token.
The completion rate varies by under $2\%$ across each ladder, so prefill is the bottleneck in every provisioned deployment, and its mean over the 5 provisioned deployments is $\mu_P$ at that input length.
The fit to the three values of $1/\mu_P$ gives the MFU and $c_a$ of Table~\ref{tab:calib}.
The fitted $\mathrm{MFU}=0.67$ corresponds to an effective compute rate of $67\%$ of GPU peak.
Ideal triangular causal attention evaluates $\ell_\mathrm{in}(\ell_\mathrm{in}+1)/2$ query--key pairs, so its quadratic coefficient approaches $2$ from above.
The fitted $c_a=2.17$ is modestly higher.
The diagonal explains less than one percent of this gap over the calibration range;
the remainder may reflect masked work in boundary tiles, softmax, and differences between the effective utilization of attention and linear layers.
We therefore read $c_a$ as an effective attention coefficient rather than an exact count of causal query--key pairs.

\paragraph{Decode.}
At each of the five input lengths, every decode instance of every provisioned deployment in the serving-capacity sweeps provides one measurement:
its logs record the number of active requests at fixed iteration intervals, so their mean over the measurement window is the batch $B$,
and $B$ divided by the rate at which the instance generates tokens over the same window is its mean iteration time.
Provisioned deployments of different shapes cover the range of batches:
when decode limits a provisioned deployment, its instances run at the batch that the KV-cache pool allows, which is nearly the same for every such provisioned deployment at a fixed input length;
when prefill limits it, its decode instances receive fewer requests and run at smaller batches.
These instances may idle between iterations, and weighting each measurement by its generation rate reduces the influence of such idle-contaminated estimates.
The three longer lengths draw on 24 provisioned deployments spanning every prefill and decode count within the 8-GPU budget, giving 60 decode instances each;
the two shorter lengths draw on 9 provisioned deployments each, giving 22, so the first regression fits 60 or 22 measurements per length and the second fits the 5 slopes.
The five slopes and intercepts give the MBU, $t_\mathrm{req}$ and $t_\mathrm{iter}$ of Table~\ref{tab:calib}.
The measured intercepts rise slightly from the shortest to the longest context length, which~\eqref{eq:etd_cal} does not represent; the mean absorbs this variation.
The fitted $\mathrm{MBU}=0.77$ corresponds to $77\%$ of peak memory bandwidth, a substantial fraction of the available bandwidth.
The two overheads matter at different scales: $t_\mathrm{iter}$ adds only a few percent to the weight-read time, whereas $t_\mathrm{req}$ accounts for a substantial part of the per-request cost at short contexts but becomes less important as the context grows.

\paragraph{Power.}
Each measurement is one GPU in one run: its power is the mean over the measurement window,
and its normalized serving throughput is the rate its instance served in that window divided by the serving capacity the model assigns to that role for that workload.
The serving-capacity sweeps and runs using the trace-derived workloads provide measurements at serving capacity, with the non-bottleneck role of each provisioned deployment below it,
while the request-rate sweep at $\ell_\mathrm{in}=4096$, which runs 7 provisioned deployments at 5 fractions of their measured serving capacity, is the only data with both roles below serving capacity at the same time.
We fit the three ramp parameters jointly by bounded nonlinear least squares.
The intercept is bounded below by the power of a GPU that holds the model weights but serves no requests, $115$\,W on this hardware, because a ramp evaluated at zero rate describes such an instance rather than an idle machine.
We use multiple starting points because some initializations place every measurement on the cap and leave $p_0$ and $\gamma$ unconstrained; we retain the fit with the lowest error.
The prefill ramp is fitted on $905$ per-GPU measurements and saturates at a normalized serving throughput of $0.99$, with a root-mean-square error of $19$\,W;
the decode ramp is fitted on $781$ and saturates at $0.50$, with $49$\,W.
The decode ramp therefore runs at its cap over the upper half of its rate range, which is why a serving-capacity error reaches the modeled power only in part.
The measurements of the fixed-length workloads and the Azure trace-derived workload are placed against the two fitted ramps in Figure~\ref{fig:ramp};
the fitted ramps approximately describe the power measurements across the calibration workloads.
In the case studies of Section~\ref{sub:exp_track}, GPUs belonging to removed instances remain idle at approximately $120$\,W each; their power is excluded from the reported total.

\subsection{Trace Request Samples}
\label{sub:trace_samples}
The trace-derived workloads are request samples drawn from the Mooncake and Azure conversation traces.
Requests whose input length exceeds a threshold are excluded first, and the experimental requests are then resampled with replacement from the remaining ones.
For Mooncake, the threshold is 38,000 input tokens, which leaves room for the output within the configured context window of 40,960 tokens; it excludes 647 of the 12,031 trace requests (5.4\%), and 1,000 requests are resampled from the rest.
For Azure, the threshold of 32,768 input tokens excludes none of the 19,366 trace requests, and 7,200 requests are resampled.
For Azure, input and output lengths are resampled as pairs, which preserves their correlation; the Mooncake runs drew the two lengths independently, so the Mooncake samples have no input--output correlation.

\section{Algorithm for Computing Serving Capacity}
\label{app:algorithm}
Algorithm~\ref{alg:capacity} summarizes the procedure for computing the serving capacity of a provisioned deployment under a given workload.
Its inputs are the workload length distribution, the provisioned deployment, the model and hardware constants of Appendix~\ref{sub:constants}, and the calibrated parameters of Appendix~\ref{sub:calibration}.
The prefill capacity follows from the compute roofline alone, whereas the decode side has to be solved jointly: the operating batch $B$ determines the decode capacity and the prefill waiting time, and the reservation occupancy they imply in turn constrains $B$.
The algorithm therefore solves the modeled memory balance of Section~\ref{sub:der_opbatch} for $B$, and the remaining quantities follow from it.
The full-pool decode capacity $\mu_D^{\max}$ is obtained by substituting the full-pool batch $B^{\max}$ of~\eqref{eq:bmax} into the decode-capacity relation~\eqref{eq:mud}; it provides the reference capacity by which the power model normalizes decode throughput in~\eqref{eq:lamn}.
In our implementation, the root of the memory balance is found with SciPy's \texttt{brentq}, an implementation of Brent's bracketed root-finding method, on the bracket $\big(0,\min(B^{\max},B^\rho)\big)$ of Section~\ref{sub:der_opbatch}, with the upper endpoint chosen inside this interval where the residual $g(B)$ is finite and positive, at the default tolerance.

The deployment-selection comparison of Section~\ref{sub:exp_e2e} applies the provisioning problem~\eqref{eq:slo} twice for each required serving capacity $\lambda_{\min}$: once with the modeled $\mu(n_P,n_D;\mathcal{W})$ and $\mathcal{P}(n_P,n_D;\mathcal{W})$, which gives the model's choice, and once with their measured counterparts $\mu^{\mathrm{meas}}$ and $\mathcal{P}^{\mathrm{meas}}$, which gives the measurement-based choice; each choice is the first deployment at or beyond $\lambda_{\min}$ along the corresponding front of Figure~\ref{fig:pareto_traces}.
The fronts of Figures~\ref{fig:pareto} and~\ref{fig:pareto_traces} are $\epsilon$-Pareto fronts~\cite{laumanns2002} with $\epsilon=3\%$ on both serving capacity and power, computed separately on the modeled and on the measured deployments; only deployments on a front are drawn, and the selection uses every evaluated deployment.
The front of Figure~\ref{fig:pareto} is computed for the fixed-length workload with 4096 input and 256 output tokens.
The model's choice is evaluated with its measured serving capacity and power.
The requirement $\lambda_{\min}$ takes 400 values spaced logarithmically between the smallest and largest measured serving capacity of the workload; the agreement is the share of these values at which the two choices coincide, and the margin of a disagreement is the relative distance from $\lambda_{\min}$ to the nearest measured serving capacity of any deployment, of which Section~\ref{sub:exp_e2e} reports the maximum over all disagreements.

\begin{algorithm}[!h]
\SetAlgoLined
\caption{$\textsc{ComputeServingCapacity}$}
\label{alg:capacity}
\SetKwFor{ForEach}{for each}{do}{end for}
\KwInput{length distribution $\mathcal{L}$; provisioned deployment $(n_P,n_D)$;
model, hardware, and serving constants $N$, $L$, $d$, $\kappa$, $\pi$, $\beta$, $C_\mathrm{tok}$, and $R$;
and calibrated MFU, $c_a$, MBU, $t_\mathrm{iter}$, and $t_\mathrm{req}$}
\KwReturn{prefill capacity $\mu_P$, full-pool decode capacity $\mu_D^{\max}$, and deployment serving capacity $\mu$}
\BlankLine
\tcc{Prefill capacity from the compute roofline}
Compute $a_P$, $b_P$, and $\mu_P$ from~\eqref{eq:tp} and~\eqref{eq:mup} \\
\BlankLine
\tcc{Decode coefficients from the memory roofline}
Compute $\bar{\ell}_{\mathrm{ctx}}$, $a_D$, and $b_D$ from~\eqref{eq:theta} and~\eqref{eq:td_mean}, with the calibrated overheads of~\eqref{eq:etd_cal} \\
\BlankLine
\tcc{Operating batch from the memory balance}
Set $CV_a^2 \leftarrow 1/n_P$ and $CV_s^2 \leftarrow \operatorname{Var}_r(t_P)/\mathbb{E}_r[t_P]^2$ \\
\ForEach{candidate $B$ selected by the bracketed root solver}{
Compute $\bar{t}_D$ and $\mu_D$ from~\eqref{eq:td_mean} and~\eqref{eq:mud} \\
Compute $\rho_P \leftarrow n_D\mu_D/(n_P\mu_P)$ and $\bar{t}_W$ from~\eqref{eq:kingman} \\
Evaluate the prefill occupancy, decode occupancy, and unused capacity of~\eqref{eq:operating_batch}
}
Set $B$ to the root at which the three occupancies sum to $C_\mathrm{tok}$,
and keep the corresponding $\bar{t}_D$ and $\mu_D$ \\
\BlankLine
\tcc{Capacities of the provisioned deployment}
Compute $\mu \leftarrow \min(n_P\mu_P,\; n_D\mu_D)$ by~\eqref{eq:twobottleneck} \\
Compute the full-pool decode capacity, with $B^{\max}$ from~\eqref{eq:bmax}
\begin{equation*}
B^{\max} \leftarrow \frac{C_\mathrm{tok}-\mathbb{E}_k[U_k]}{\bar{\ell}_{\mathrm{ctx}}+R},
\quad
\mu_D^{\max} \leftarrow \frac{B^{\max}}{\mathbb{E}_r[\ell_\mathrm{out}-1]\,(a_D + b_D B^{\max})}
\end{equation*} \\
\Return $(\mu_P, \mu_D^{\max}, \mu)$
\end{algorithm}

\section{Sequential Scale-Down Experiment}
\label{app:shrink}
We extend the two provisioning scenarios of Section~\ref{sub:exp_track} with a sequential scale-down experiment at a fixed request rate.
The experiment illustrates the boundary between efficient scale-down and overload: reducing provisioned capacity can lower power while the resulting deployment retains sufficient serving capacity, but once serving capacity falls below the request rate, further power reduction causes sustained queueing and rapidly increasing latency.
The deployment follows $3\mathrm{p}2\mathrm{d}\to3\mathrm{p}1\mathrm{d}\to2\mathrm{p}1\mathrm{d}\to1\mathrm{p}1\mathrm{d}$ at $3.6$\,req/s, or $80\%$ of the serving capacity of 2p1d (Figure~\ref{fig:scn_shrink}).
The first two resulting deployments remain on the Pareto front and retain sufficient serving capacity for this request rate.
Removing a decode instance reduces total power by approximately $520$\,W and increases median TPOT from $24$ to $29$\,ms, while median TTFT remains unchanged.
Removing a prefill instance next reduces total power by approximately $180$\,W, less than the decode removal because the two remaining prefill instances now operate at $80\%$ of their serving capacity.
Median TTFT increases from $0.5$ to $0.9$\,s, with individual-request TTFT reaching up to $6$\,s.

After the third removal, 1p1d has a serving capacity of $2.3$\,req/s, below the fixed request rate of $3.6$\,req/s.
Serving throughput falls to the serving capacity, while requests accumulate in the queue because the request rate exceeds it.
TTFT increases throughout the overloaded period, reaching a maximum of $320$\,s by the end of the run, approximately $14$\,min later.
The $470$\,W power reduction therefore comes at the cost of a persistent throughput shortfall and a growing queue, so 1p1d cannot sustain the offered request rate.
Thus, the final removal marks a qualitative change in operating regime: the earlier power savings preserve a stable deployment, whereas the final reduction is obtained only by leaving serving capacity below the request rate.

\begin{figure}[H]
  \centering
  \includegraphics[width=\columnwidth]{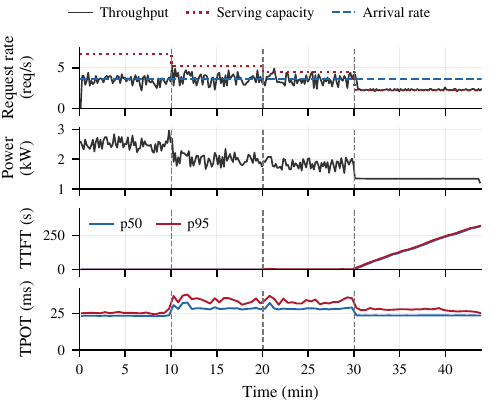}
  \caption{Three removals chained under one request rate
    ($3\mathrm{p}2\mathrm{d}\to3\mathrm{p}1\mathrm{d}\to2\mathrm{p}1\mathrm{d}\to1\mathrm{p}1\mathrm{d}$),
    the same rows as Figure~\ref{fig:scn_pair}.}
  \label{fig:scn_shrink}
\end{figure}

\end{document}